\documentclass[11pt]{article}

\usepackage[margin=1in]{geometry}
\usepackage{amsmath,amssymb,amsthm,mathtools}
\usepackage{algorithm}
\usepackage{algpseudocode}
\usepackage{microtype}
\usepackage{enumitem}
\usepackage[colorlinks=true,linkcolor=blue,citecolor=blue,urlcolor=blue]{hyperref}\usepackage{xcolor}
\usepackage{amsthm}
\usepackage{cleveref}
\usepackage{aliascnt}
\usepackage{bbm}
\usepackage{url}
\usepackage{bm}
\usepackage{comment}
\usepackage[sort&compress,round]{natbib}
\usepackage[titletoc]{appendix}
\usepackage[nottoc]{tocbibind}
\usepackage{placeins}
\allowdisplaybreaks

\algblockdefx[WHENEVER]{Whenever}{EndWhenever}
  [1]{\textbf{whenever} #1 \textbf{do}}
  {\textbf{end whenever}}

\newtheorem{theorem}{Theorem}[section]

\newaliascnt{lemma}{theorem}
\newtheorem{lemma}[lemma]{Lemma}
\aliascntresetthe{lemma}
\crefname{theorem}{theorem}{theorems}
\Crefname{theorem}{Theorem}{Theorems}
\crefname{lemma}{lemma}{lemmas}
\Crefname{lemma}{Lemma}{Lemmas}

\newaliascnt{fact}{theorem}
\newtheorem{fact}[fact]{Fact}
\aliascntresetthe{fact}
\crefname{fact}{fact}{facts}
\Crefname{fact}{Fact}{Facts}

\newaliascnt{corollary}{theorem}
\newtheorem{corollary}[corollary]{Corollary}
\aliascntresetthe{corollary}
\crefname{corollary}{corollary}{corollaries}
\Crefname{corollary}{Corollary}{Corollaries}

\crefname{informaltheorem}{informal theorem}{informal theorems}
\Crefname{informaltheorem}{Informal Theorem}{Informal Theorems}

\usepackage{thmtools,thm-restate}

\newenvironment{proofsketch}
  {\begin{proof}[Proof sketch]}
  {\end{proof}}

\makeatletter
\newcommand{\citeinorder}[1]{%
  \begingroup
  \def\NAT@sort{0}%
  \def\NAT@cmprs{0}%
  \cite{#1}%
  \endgroup
}
\makeatother

\newcommand{\rate}{\operatorname{rate}}

\newcommand{\opt}{\mathrm{Opt}}
\newcommand{\alg}{\mathrm{Alg}}
\newcommand{\name}{\textsc{ThresholdBelady}}
\newcommand{\manpol}{\textsc{ManThresholdBelady}}
\newcommand{\namepr}{\textsc{GuardedEviction}}

\newcommand{\expect}{\operatorname{E}}
\newcommand{\error}{\mathcal{E}}

\newcommand{\lopt}{\mathrm{L}_\mathrm{opt}}
\newcommand{\lman}{\mathrm{L}_\mathrm{man}}
\newcommand{\lghost}{\mathrm{L}_\mathrm{gst}}
\newcommand{\pname}{\mathrm{SWC}}
\newcommand{\pnameit}{\mathit{SWC}}
\newcommand{\ohp}{\textsc{ManGuardedEviction}}
\newcommand{\manpar}{\textsc{ManPartition}}
\newcommand{\manfinal}{\textsc{MandatoryAdmission}}
\newcommand{\greedypart}{\textsc{GreedyPartition}}
\newcommand{\errorname}{\textsc{ErroneousPredictions}}
\newcommand{\robustname}{\textsc{RobustSwap}}
\newcommand{\con}{\textsc{Cont}}
\newcommand{\cert}{\textsc{R}}

\newcommand{\pol}{\mathcal{P}}
\newcommand{\cost}{\mathrm{cost}}

\newcommand{\dimitris}[1]{\textcolor{blue}{[Dimitris: #1]}}

\title{Stochastic Weighted Caching with Per-Request Predictions}
\author{
Shuchi Chawla \smallskip \\\medskip shuchi@cs.utexas.edu \\ UT Austin
\and 
Dimitris Christou\smallskip  \\\medskip christou@cs.utexas.edu \\ UT Austin
}

\date{}

\begin{document}

\maketitle
\thispagestyle{empty}
\begin{abstract}
Both the unweighted and weighted variants of caching are cornerstone problems in competitive analysis, with tight randomized competitive ratios of $\Theta(\log k)$, where $k$ is the cache size. We study weighted caching with per-request predictions, where each page request is accompanied by a prediction of the next arrival time of the same page. For unweighted caching, such perfect predictions suffice to implement Belady’s offline optimal rule. For weighted caching, however, previous work shows that a logarithmic lower bound persists even with perfect predictions on adversarial request sequences. Stochastic assumptions alone are also insufficient: the classical caching lower bounds already apply to uniform independent requests. In this work, we ask whether combining these two sources of additional information can overcome this barrier, and answer this question affirmatively.

Specifically, we develop an efficient 
cache eviction policy that achieves a competitive ratio of $15.75$ using perfect per-request predictions when page requests are drawn independently from an arbitrary, known static distribution in every step, and page costs are arbitrary. Underlying our analysis is a novel lower bound on the optimal cost obtained through an ex-ante relaxation of the cache's covering constraints. Our algorithm's performance degrades gracefully under erroneous predictions: we prove a bound of $O(\min\{1+\mathcal E,\log k\})$ on the competitive ratio, where $\mathcal E$ aggregates each page’s average prediction error, and may be unknown to the algorithm.
    
\end{abstract}
\setcounter{tocdepth}{2}
\tableofcontents
\thispagestyle{empty}

\newpage
\setcounter{page}{1}
\section{Introduction}

Caching (also known as paging) is a foundational problem in online algorithms and competitive analysis. In the most basic version, an online algorithm, or cache eviction policy, serves requests from a universe of $n$ pages using a cache of capacity $k<n$. Each time a requested page is absent from the cache, it must be fetched at unit cost; if the cache is full, the policy chooses which page to evict. In the weighted version, fetching page $p$ incurs a page-specific cost $w_p\ge 0$, modeling differences in retrieval cost arising, for example, from where pages reside in a storage hierarchy. Caching and its variants have attracted sustained attention for their widespread use in computer systems, their role as special cases of central online problems such as $k$-server and metrical task systems, and their value as testing grounds for new techniques in online algorithms (e.g. \citeinorder{buchbinder2009,buchbinder2014,bubeck2018}).

The optimal worst-case competitive guarantees for both unweighted and weighted caching are well understood. For unweighted caching, \cite{fiat1991} obtained an $O(\log k)$ competitive ratio through the randomized marking algorithm and established a matching lower bound up to constant factors. The lower bound, in particular, holds even for the special case of $n=k+1$ pages and with requests drawn \textit{independently and uniformly at random} from the universe. For weighted caching, obtaining a sublinear competitive ratio remained an open problem for many years, until \cite{bansal2012} developed an $O(\log k)$-competitive randomized algorithm using an online primal-dual approach, achieving the optimal worst-case ratio up to constant factors and closing the gap between the weighted and unweighted versions. 
The $\Omega(\log k)$ lower bound, together with the problem's practical importance, motivates understanding what additional structure or information allows eviction policies to achieve constant competitiveness.


Learning-augmented algorithms offer one approach to overcoming these worst-case barriers by giving online algorithms predictions about future inputs. In pioneering work, \citet{lykouris2018} demonstrated the power of this approach for caching using a simple prediction model: whenever a page is requested, the algorithm receives a prediction of that page's next request time. These \emph{per-request predictions} (PRP) are natural in practice, as a process requesting a page may be able to estimate when it will next need that page. For unweighted caching, perfect PRP suffice to implement Belady's optimal eviction rule~\citep{belady1966}, reducing the competitive ratio from $\Theta(\log k)$ to $1$. Unfortunately, the power of this prediction model does not extend to weighted caching in the worst case: even with \textit{perfect per-request predictions}, every randomized online algorithm has competitive ratio $\Omega(\log k)$~\citep{jiang2020}. 

In this work, we show that \textbf{input stochasticity and per-request predictions together reduce the competitive ratio for weighted caching from $\bm{\Theta(\log k)}$ to $\bm{O(1)}$.} 

We study \emph{Stochastic Weighted Caching} (SWC), where pages have arbitrary fetching costs $w_p\ge 0$ and independent Poisson request processes with known, arbitrary rates $\mu_p>0$. At each request, the algorithm receives a prediction of that page's next arrival time. We measure competitiveness using long-run expected cost per unit time. With perfect predictions, our algorithm is $O(1)$ competitive against the offline optimum, independent of the cache size, the number of pages, and the range of weights and arrival rates. Crucially, neither stochasticity nor predictions alone suffice: as discussed, the logarithmic lower bound persists both under stochastic requests without predictions and under adversarial requests with perfect per-request predictions.

\begin{restatable}{theorem}{perfectpredictiontheorem}
\label{thm:perfect}
    Any instance of $\pnameit$ admits an efficient deterministic online caching policy (\Cref{alg:mandatory-final}) that is $15.75$-competitive when supplied with perfect per-request predictions.
\end{restatable}

Observe that a caching instance can be specified by the order of requests alone, since fetching costs do not depend on the time elapsed between requests. We nevertheless use the continuous-time Poisson model because it allows our algorithm to cleanly exploit the independence of the arrival processes for different pages. 
Moreover, the Poisson model is equivalent to the static distribution model at the level of request order: after discarding arrival times, successive requests are independent draws from a fixed distribution, with page \(p\) chosen with probability proportional to $\mu_p$ (see Appendix~\ref{app:discussion}). Although our algorithm assumes this distribution is known, it uses only parameters that can be efficiently estimated from observations. 

The Poisson model also naturally captures settings in which requests for different pages originate from separate, independent sources. This interpretation motivates our per-request prediction model: when issuing a request, a source may be able to predict the time of its own next request without knowing which requests other sources will issue in the meantime. Our model requires only this prediction of the requested page's next arrival time. By contrast, prior work showed that constant-competitive guarantees for adversarial weighted caching require a much stronger prediction model that reveals, at each request to a page $p$, the entire sequence of requests preceding $p$'s next arrival~\citep{jiang2020}.




Unlike in unweighted caching, it is not immediately clear how to use per-request predictions in the stochastic weighted setting. Should the algorithm evict the page whose next request is furthest in the future, or measure this delay relative to the page's expected interarrival time? How should the eviction decision depend on the page's weight? We give a high-level overview here, leaving the precise rules and treatment of corner cases to the formal description. At the heart of our algorithm is a subroutine called \emph{Guarded Eviction}, which handles the special case $n=k+1$, where the cache can hold all but one page. The subroutine uses a weight threshold and a lateness threshold. Pages whose weights exceed the weight threshold remain in the cache regardless of their next arrival times. Among the remaining pages, a page is evicted regardless of its weight if the time until its next request exceeds its expected interarrival time by more than the lateness factor. To extend this approach to arbitrary $n$, our algorithm solves a relaxation that estimates how often each page is evicted by the hindsight optimum. It uses the resulting solution to partition the cache and the page universe into components, each of which can independently run Guarded Eviction. In the next subsection, we explain why this approach achieves a constant competitive ratio, and discuss the main challenges in the analysis.

We turn our attention next to the more realistic setting of erroneous predictions. Following the learning-augmented caching literature, we measure the error of each prediction as the absolute difference between the predicted and actual next arrival times. We allow an adversary with full knowledge of the realized request sequence to choose predictions arbitrarily. Because requests arrive in continuous time and our performance guarantees hold in expectation over the arrival process, we express the adversary's error budget as an expected error rate $\error$: over any long-enough time horizon $T$, the sum of prediction errors, taken in expectation over the realized request sequence, is at most $\error T$. Subject to this constraint, the adversary may allocate error arbitrarily across realized request sequences and individual requests.

Our algorithm's parameters, partitions, and thresholds depend only on the known page weights and arrival rates, so the same algorithm applies without modification. We show that its performance degrades gracefully as prediction error increases and obtain the following result.


\begin{restatable}{theorem}{erroneouspredictiontheorem}
\label{thm:erroneous}
    Any instance of $\pnameit$ admits an efficient deterministic online caching policy (\Cref{alg:error-final}) that is
    $O(\min\{1+\error,\log k\})$-competitive when supplied with per-request predictions of total error $\error$.
\end{restatable}


Crucially, our policy does not need to know the prediction error in advance. Furthermore, it satisfies all of the desired properties of learning-augmented algorithms: its competitive ratio is constant for accurate predictions (consistency), grows
at most linearly with their error (smoothness), and remains logarithmic in the cache size regardless of their quality (robustness).

Finally, one might wonder whether stochasticity coupled with per-request predictions is so powerful as to leave a vanishing gap---$1+o(1)$---between the online algorithm and the offline optimum. We show that this is not the case: every online algorithm has a competitive ratio bounded away from $1$ by an absolute constant. See Appendix~\ref{app:lb}.

\begin{restatable}{theorem}{lowerbound}
\label{thm:lowerbound}
    There exists an instance of $\pnameit$ such that the competitive ratio of any online caching policy with perfect per-request predictions is at least $1.001$. 
\end{restatable}

\subsection{Technical Overview}

We measure the performance of a caching policy $\alg$ in terms of its asymptotic expected cost rate, defined as
\[
    \rate(\alg):=\limsup_{T\to\infty}
    \frac{\operatorname{E}[\operatorname{cost}_{\alg}(T)]}{T},
\]
where $\operatorname{cost}_{\alg}(T)$ is the total fetching cost incurred by $\alg$ on the (random) instance containing all requests up to some time $T\ge 0$. We say that $\alg$ is $c$-competitive if $\rate(\alg)\le c\cdot\rate(\opt)$, where $\opt$ denotes the offline optimum over the realized sequence.

\subsubsection*{Intuition behind the constant factor: the case of a single  hole}
To understand why the gap between online and offline policies is constant in our setting, consider the special case $n=k+1$. In this setting, a full cache leaves exactly one page outside; we call this page the \emph{hole}. For this thought experiment, we also assume \emph{optional admission}: a policy must pay the cost of a cache miss, but need not bring the requested page into the cache.

Without predictions, memorylessness implies that whenever page $p$ is the hole, the instantaneous expected cost rate is $\mu_p w_p$. Thus, the optimal online policy without predictions simply keeps a page minimizing this quantity permanently outside the cache, achieving a long-run cost rate of $\min_p \mu_p w_p$. We call $q_p :=w_p\mu_p$ the {\em intrinsic cost} of page $p$. This quantity provides a common scale for comparing pages with different fetching costs and arrival rates. 

How can perfect predictions improve on this baseline? One possibility is to exploit intervals during which a page remains unrequested for much longer than its mean interarrival time. Consider a \emph{Threshold Belady} rule with a parameter $\theta>0$. On a cache miss, after receiving the requested page's new prediction, the policy looks for a page $p$ whose next request is at least $\theta/\mu_p$ time units away and makes any such page the hole. If no page qualifies, it leaves the requested page outside the cache. Whenever a qualifying page $p$ is chosen, the interval until the next miss has length at least $\theta/\mu_p$ and ends with a cost of $w_p$. Amortizing this cost over the interval gives a rate of at most $w_p\mu_p/\theta = q_p/\theta \le \max_p q_p/\theta$. We would therefore like to choose $\theta$ as large as possible while keeping failures sufficiently rare. At any fixed time, the probability that page $p$ does not return for another $\theta/\mu_p$ units of time is exactly $e^{-\theta}$ since the next arrival time of $p$ is distributed as $\mathrm{Exp}(\mu_p)$. So choosing $\theta\approx\log n$ suffices to control the probability and the cost of failures.

This rough calculation already provides a hint that predictions can allow an online algorithm to overcome the logarithmic barrier for the competitive ratio. When all pages have comparable intrinsic costs, the above bound improves on the optimal cost without predictions by a logarithmic factor. With unequal intrinsic costs, however, the dependence on $\max_p q_p$ can be too expensive: even a long interval without requests may not justify leaving a page of high intrinsic cost outside the cache. A natural refinement is to protect the expensive pages and apply the Threshold Belady rule only among the cheaper ones.

In particular, order the pages in non-decreasing order of intrinsic costs: $q_1\le\cdots\le q_n$, and let $P_m$ denote the prefix of size $m$ over this ordering. Apply the Threshold Belady rule to $P_m$ for $m\ge2$ with $\theta_m=(1/2)\log m$, and keep the remaining $n-m$ pages permanently in cache. The preceding calculation gives a cost rate of $O(q_m/\theta_m)$. For $m=1$, simply leave page $1$ permanently outside, at rate $q_1$. Choosing the best prefix yields a cost rate of $O(\alpha)$, where
\[
    \alpha:=\min_{1\le m\le n}\frac{q_m}{\log(m+1)}.
\]
We call this policy {\em Guarded Eviction}.

Why can the offline optimum not do asymptotically better? For a hole interval to have amortized cost rate below $\alpha'$, the offline optimum needs to find a page $p$ that will remain unrequested for a $w_p/\alpha'$ amount of time. If $\alpha'$ is too small, the likelihood that such a page exists under the stochastic arrival process will be small. For example, if {$\alpha>0$ and $\alpha'=\alpha/2$, then at any fixed time the probability that page $p$ will return after $w_p/\alpha'$ units of time is at most $e^{-w_p\mu_p/\alpha'} = e^{-2q_p/\alpha}$. Since $\alpha\le q_p/\log(p+1)$, this probability is at most $e^{-2\log(p+1)}=1/(p+1)^2$. In other words, the probability that the offline optimum can find {\em any} page $p$ that incurs a cost rate below $\alpha/2$ until the next miss is at most
\[
    \sum_p (p+1)^{-2}<\frac{2}{3}.
\]
Therefore, this suggests that for a constant fraction of the misses, the optimum will incur a cost rate of $\Omega(\alpha)$---only a constant factor smaller than the online cost.

\subsubsection*{A formal lower bound on the optimum}

The preceding calculation does not yet constitute a proof: whether the offline optimum can find a page with a low amortized cost rate is correlated with the time until its next miss. One of our main technical contributions is a simple lower bound on the optimal cost rate expressed solely in terms of the pages' intrinsic costs. We obtain this bound through an \emph{ex-ante relaxation} of the policy's covering constraint. A feasible policy must ``cover'' the hole by keeping at least one page outside the cache at every point in time. The relaxation imposes this requirement only on average: the expected number of pages outside the cache, averaged over time, must be at least one. Thus, overlapping coverage at some times may compensate for uncovered periods at others.

Suppose a policy misses a fraction $f_p$ of page $p$'s requests. Each miss allows the page to have been outside the cache during at most the preceding interarrival gap.
Even with complete knowledge of the future, the most time the policy could obtain from these misses comes from selecting the longest $f_p$ fraction of its gaps. For the Poisson arrival process, these are precisely the gaps exceeding $\log(1/f_p)/\mu_p$, whose average length is therefore $(1+\log(1/f_p))/\mu_p$ due to memorylessness. Since the selected gaps occur at rate $\mu_pf_p$, they cover at most a
    \[
        \Phi(f_p):=f_p\bigl(1+\log(1/f_p)\bigr)
    \]
fraction of the timeline in expectation. The fractions $f_p$ must therefore satisfy $\sum_p \Phi(f_p)\ge 1$. Minimizing the policy's cost rate of $\sum_pq_pf_p$ subject to the constraint provides a lower bound.

Some algebra then shows that for every instance $(\mu_p,w_p)_{p\in[n]}$, the prefix benchmark $\alpha$ is at most a constant times the value of this relaxation. Together with the bound of $O(\alpha)$ on the cost rate of Guarded Eviction, this establishes the desired constant-factor gap.

\subsubsection*{The case of general $\bm{n}$}

The special case of $n=k+1$ is easy to generalize to arbitrary $n$ for the case of optional admission. Let $r=n-k$ denote the number of holes. We show that using the optimal solution to the ex ante relaxation, and in particular the covering constraint $\sum_p \Phi(f_p)\ge r$, we can greedily decompose the set of all pages into $r$ subsets in a manner that preserves the lower bound up to a constant factor. We then run the one-hole policy on each subset independently, generating all required holes. 

\subsubsection*{Mandatory admission} 

We now consider the standard caching model where requested pages must be brought into cache. We follow the same high-level
approach as for optional admission, but must now account for the additional evictions forced by admitting every requested page. This entails strengthening the lower bound, adapting the guarded eviction policy and modifying the partition. In particular:
\begin{enumerate}
    \item The lower bound described above can substantially underestimate the optimal cost. For example, $r$ pages of weight zero can supply all $r$ units of coverage at no cost in the relaxation. Under mandatory admission, however, their requests force other pages out of the cache, potentially generating positive cost when those pages return. To capture this interaction, we introduce an auxiliary \emph{ghost page} representing the combined request stream of a selected set of $r$ pages. We then impose an additional balance constraint between the ghost's miss rate and the aggregate miss rate of the remaining pages; this provides a strengthened lower bound. 
    \item Recall that in the one-hole special case, Guarded Eviction responds to a miss by evicting a page $p$ whose next request is at least $\theta/\mu_p$ time units away. Under optional admission, if the only qualifying page is the requested page, we simply leave it out of the cache. Under mandatory admission, we instead use a pair of \emph{buffer pages} as default eviction choices. We show that a careful choice of these pages bounds the additional eviction cost.
    \item Finally, the partition into single-hole instances used for optional admission may contain singleton components, whose pages are permanently left out of the cache. Under mandatory admission, we combine all these cheap singleton pages into a single ghost page and place it together with other pages in a special component with one hole. We then charge the eviction cost incurred in this component to the optimum through the balance constraint introduced in our strengthened lower bound above.
\end{enumerate}

\subsubsection*{Erroneous predictions}


Finally, we relax the assumption of perfect predictions. We first separate robustness from the task of exploiting accurate predictions. An $O(\log k)$ guarantee is already possible without predictions \citep{bansal2012}. A standard black-box combination \citep{blum97} of two caching policies achieves the better of their guarantees up to a constant factor. Combining a prediction-based policy with this prediction-free alternative therefore provides robustness automatically. It remains to design a policy whose competitive ratio grows as $O(1+\mathcal E)$ with the prediction error.

Crucially, the lower bounds and partitions developed above depend only on the stochasticity of the model, so they remain valid under erroneous predictions. 
The only step that uses predictions is the implementation of the component policies. We can therefore retain the decomposition and focus on how prediction errors affect the guarded eviction that implements the one-hole policy over each component. There are two ways an incorrect prediction can hurt the policy:
\begin{enumerate}
    \item First, a page predicted to return late may arrive earlier than expected, shortening the miss-free interval obtained by evicting it. This loss is directly bounded by the prediction error: the promised waiting time is at most the actual waiting time plus the amount of overprediction.


    \item Second, a page predicted to return early may actually remain unrequested for a long time, causing the policy to overlook a useful eviction opportunity. The key observation is that 
    the stochastic model typically provides many pages with long actual gaps. For a search to overlook all of them, many predictions must therefore be substantially wrong at the same time. This allows us to distribute the cost of a failed search among many errors. Although the same predictions may cause repeated failures, the stochastic model also controls how frequently these failures occur, allowing us to bound their total cost by the total prediction error.
\end{enumerate}

\subsection{Roadmap} 
We formally define $\pname$, our asymptotic notion of competitive ratio, and the prediction model in \Cref{sec:preliminaries}. In \Cref{sec:optional}, we develop an online algorithm for the (easier) optional admission model with perfect predictions, proving a simpler version of \Cref{thm:perfect}. This setting allows us to convey the key technical components of our main result without the complexities of mandatory admission. In \Cref{sec:mandatory} we consider the standard (mandatory admission) caching model with perfect per-request predictions and give the proof of \Cref{thm:perfect}. Finally, we discuss erroneous predictions and prove \Cref{thm:erroneous} in \Cref{sec:predictions}. 


\subsection{Related Work}\label{sec:related}

The competitive analysis of caching was pioneered by \cite{sleator1985}, who established the optimal deterministic competitive ratio of $k$ for unweighted paging, achieved by policies such as LRU. For randomized algorithms against an oblivious adversary, \cite{fiat1991} gave a $2H_k$-competitive marking algorithm and an $H_k$ lower bound, subsequently matched by \cite{mcgeoch1991}. For weighted paging, \cite{chrobak1991} obtained an optimal $k$-competitive deterministic algorithm, and \cite{young1994} developed the primal-dual GreedyDual algorithm with the same guarantee. The optimal randomized dependence on the cache size was settled by \cite{bansal2012}, who obtained an $O(\log k)$-competitive algorithm. Beyond its many practical applications, weighted caching is a special case of the $k$-server problem~\citep{manasse1990}, through its representation on a weighted star, and of metrical task systems~\citep{borodin1992}, with cache configurations as states. 

\paragraph{Stochastic Caching.}
Earlier approaches to capturing structure in caching inputs include the access-graph model of \cite{borodin1995}, modeling locality of reference, and the framework of \cite{koutsoupias2000}, restricting the input distributions available to the adversary. Probabilistic models of page requests have been studied since the work of \cite{aho1971}. 

Beyond independent requests, \cite{karlin2000} study paging when requests follow a known Markov chain. Against the optimal online policy with the same distributional knowledge, \cite{lund1999} give a $4$-competitive dominating-distribution algorithm; \cite{pabbaraju2025} sharpen its analysis to obtain a factor of $2$ for Markov paging.

\cite{fairstein2024} study weighted paging with advance knowledge of independent, potentially non-identical request distributions. They obtain constant-factor guarantees against the optimal online policy, both with full distributional information and under a limited-horizon model. Their work is closely related in its use of probabilistic information, but differs from ours in both the prediction model and the benchmark. 

Against the clairvoyant offline optimum, however, even independent uniform requests yield an $\Omega(\log k)$ lower bound~\citep{fiat1991}. Recently, \cite{kumar2026} introduce \emph{subset entropy}, a distribution-specific parameter that yields finer competitive guarantees for unweighted caching. Our results pursue a complementary direction: we allow arbitrary request frequencies and fetching costs, and use predictions of realized return times to obtain a constant competitive ratio against the offline optimum, including for distributions that are hard without predictions. 

\paragraph{Adversarial Caching with Per-Request Predictions.}
\cite{lykouris2018} introduced learning-augmented unweighted caching with predictions of each requested page's next occurrence, obtaining guarantees that improve with prediction accuracy while retaining $O(\log k)$ robustness. \cite{rohatgi2020} improved the dependence on prediction error, and \cite{wei2020} showed that following predicted Belady evictions and combining this policy with a robust online algorithm gives improved guarantees through a simpler algorithm. \cite{skachkov2025} further sharpen the analysis of this prediction-following policy. Recent work further improves the consistency--robustness trade-off \citep{chen2025guard,chen2026paging}. These results concern unweighted caching, where perfect per-request predictions suffice to implement Belady's optimal eviction rule~\citep{belady1966}.

For weighted caching, even perfect per-request predictions leave competitive-ratio lower bounds of $\Omega(k)$ for deterministic algorithms and $\Omega(\log k)$ for randomized algorithms~\citep{jiang2020}. To overcome this barrier, \cite{jiang2020} introduce \emph{strong per-request predictions}, which reveal the entire request sequence up to the next occurrence of the current page, and obtain a $2$-competitive algorithm. Using ordinary return-time predictions, \cite{bansal2022} instead obtain tight guarantees in terms of the number $\ell$ of weight classes: an $\ell$-competitive deterministic algorithm and an $O(\log\ell)$-competitive randomized algorithm under perfect predictions, together with guarantees that degrade gracefully with prediction error. Thus, these predictions can remove dependence on the cache size when the weights have limited diversity. Our stochastic assumption permits a constant guarantee independent of both the cache size and the number of weight classes.


\paragraph{Future Information in Stochastic Optimization.}
Another work close in spirit to ours is \cite{christou2026}, who study online batching with finite lookahead. They show that stochastic arrivals alone and perfect finite lookahead alone each leave a competitive-ratio lower bound bounded away from $1$. Combining the two, however, yields a competitive ratio that approaches $1$ exponentially fast as the lookahead grows. Our work establishes a similar benefit from combining stochasticity and predictions in weighted caching: although either assumption alone leaves a logarithmic barrier, together they yield a constant competitive ratio using per-request predictions.

\section{Preliminaries}\label{sec:preliminaries}

An instance of Stochastic Weighted Caching ($\pname$) consists of a set of $n\ge 1$ pages and a cache of size $k\le n$. Each page $p\in [n]$ has a fetching cost $w_p\ge 0$ and a (Poisson) arrival rate $\mu_p>0$. Starting from time $t=0$, pages arrive independently based on Poisson processes with their respective rates; that is, the next arrival time of page $p\in [n]$ after time $t$ is $t'=t+X_p$ where $X_p\sim\mathrm{Exp}(\mu_p)$ independently of other pages and the history up to time $t$ (memorylessness). We also use $r:=n-k\ge 0$ to denote the number of pages that do not fit in the cache. Finally, we will refer to the product $q_p:=\mu_p w_p$ as the intrinsic cost of page $p$. All logarithms are natural.

\subsubsection*{Caching Policies} 
A caching policy maintains at every time $t$ a set of cached pages $C(t)\subseteq [n]$\footnote{We allow an arbitrary initial cache configuration at $t=0$. Changing this configuration to another one incurs only a finite initialization cost and therefore does not affect the asymptotic cost rate defined below.} such that $|C(t)|\le k$. When a page $p$ arrives at time $t$, we say that it is a hit if $p\in C(t)$ and a miss if $p\notin C(t)$; upon a miss, the policy must pay the fetching cost $w_p$ of the page. Under \textit{mandatory admission}, a requested page must belong to the cache immediately after its request is served. Thus, if the cache is full and the request is a miss, the policy must evict another page. Under \textit{optional admission}, the policy may instead serve the request at cost $w_p$ without placing the requested page in the cache. As is standard in the caching literature, we will assume without loss that all policies are \textit{lazy}, meaning that they only evict a page $p$ upon a miss.

The following equivalent (standard) covering interpretation of caching will be useful throughout our work. We call a page outside the cache a \textit{hole}, and use $H(t):=[n]\setminus C(t)$ to denote the set of holes maintained by some policy at time $t$. At every point in time, the policy must maintain $r$ holes (i.e. $|H(t)|\ge r$) to ensure feasibility. 

\subsubsection*{Asymptotic Cost Rate}
Given a caching policy $\pol$, we use $\cost_\pol(T)$ to denote the total cost paid by it up to time $T\ge 0$; note that this is a random variable that depends both on the random sequence of page arrivals up to time $T$ as well as any internal randomness that $\pol$ might use. We evaluate the performance of a policy in terms of its average asymptotic cost rate, formally defined as
\[\rate(\pol):=\limsup_{T\to\infty}\frac{\expect [\cost_\pol(T)]}{T} = \limsup_{T\to\infty}\frac{\sum_{p\in [n]}w_p\cdot \expect [\text{\# of misses on $p$ before $T$}]}{T} .\]
The supremum is taken to ensure that this quantity is well-defined for all policies $\pol$. Since page arrival rates and fetching costs are finite, $\cost_\pol(T)$ grows linearly with $T$ for any lazy policy and therefore the rate is well-defined.

We use $\rate(\opt)$ to denote the rate of the hindsight optimal policy that, for every fixed $T$, observes realizations over the entire horizon before determining how to serve them.
We say that a caching policy $\pol$ is $c$-competitive for some $c\ge 1$ if $\rate(\pol)\le c\cdot \rate(\opt)$.

\subsubsection*{Per-Request Predictions and Error Model}
A (purely) online caching policy initially only knows the fetching costs and arrival rates of the pages, observes the requests as they arrive and must serve them immediately, making irrevocable caching decisions. In the per-request prediction model (PRP), whenever page $p$ gets requested for the $i$-th time at time $t_p(i)$, the policy is also given a \textit{prediction} of its next arrival time $t'_p(i+1)$.\footnote{We can assume without loss that at $t=0$ a prediction of each page's first arrival is also given to the policy. Otherwise, we can wait for a (finite expectation) initialization period until such predictions have been obtained, implementing an arbitrary caching policy in the meantime, without affecting the asymptotic cost rate.} We say that the predictions are \textit{perfect} if $t'_p(i+1)=t_p(i+1)$ always, i.e. the prediction always matches the true realization of the random variable $t_p(i+1)=t_p(i) + \mathrm{Exp}(\mu_p)$. Otherwise, we say that the predictions are \textit{erroneous}.

Under erroneous per-request predictions, the values $t'_p(i+1)$ are completely arbitrary. We adopt the \textit{almighty adversarial model} where an adversary knows in advance the algorithm's description and internal randomness, as well as the realized request sequence, and can use this information to modify the actual arrival times $t_p(i)$ in an arbitrary manner to obtain the predictions $t'_p(i)$. For any such mapping, we define the average asymptotic prediction error of page $p$ as
\[
    \epsilon_p
    := \limsup_{T\to \infty}\frac{\expect \big[
    \sum_{i: \; t_p(i)\le T}
    \left|t_p(i)-t'_p(i)\right|\big]}{T},
\]
and the aggregate prediction error as $\error :=\sum_{p\in [n]}\epsilon_p$. Importantly, the policies we design do not require advance knowledge of the error $\error$, nor the process that generates it, but their performance automatically scales with it.

\section{The Optional Admission Setting with Perfect Predictions}\label{sec:optional}

We begin with the optional admission variant of $\pname$ under perfect per-request predictions. Recall that unlike mandatory admission, the optional admission model allows a missed request to be served from outside the cache. Our aim is to develop the key intuition behind our main result (\Cref{thm:perfect}) without the additional technical complications of mandatory admission. To that end, we emphasize a modular analysis and clarity of exposition over optimizing the competitive ratio we achieve. Our main result is the following:

\begin{theorem}\label{thm:const_opt}
    Under optional admission, any instance of $\pnameit$ admits an efficient deterministic online caching policy (\Cref{alg:optional}) that is $42$-competitive when supplied with perfect per-request predictions.
\end{theorem}

Following our discussion from the introduction, we prove this result in three steps. First, we derive a lower bound on the cost rate of the offline optimum. Then, we design and analyze a constant-competitive policy for one-hole instances with $n=k+1$. Finally, we reduce the general case of arbitrary $n$ to a collection of one-hole instances. Together, these ingredients lead to the design of \Cref{alg:optional} and the proof of \Cref{thm:const_opt}.

\subsection{A Lower Bound on the Offline Optimal Rate}

We derive a lower bound on the asymptotic cost rate of any feasible policy, and hence of the hindsight optimum. Fix any set of pages $S\subseteq[n]$ that must be served using $k$ cache slots. At every point in time, at least $|S|-k$ pages must be holes. We relax this requirement by requiring that the number of holes is at least $|S|-k$ {\em on average} over time. Furthermore, for a page $p$ evicted between its $i$-th and $(i+1)$-th arrival, we will assume that the page is a hole for the entire duration of the interval $(t_p(i), t_p(i+1)]$. This improves the algorithm's coverage without changing its cost. 


With these changes, suppose that on average, a policy evicts a page $p\in S$ during an $f_p$-fraction of its interarrival intervals. Each eviction incurs cost $w_p$. So, no matter how these intervals are selected, the policy incurs cost at rate $(\mu_p f_p)\cdot w_p = f_p q_p$. On the other hand, the total hole coverage provided by $p$'s eviction is maximized by evicting $p$ during its longest interarrival intervals. Since these intervals are distributed as $\mathrm{Exp}(\mu_p)$, this corresponds to selecting those of length at least $\tau_p$, such that
\[ \mathrm{Pr}[\mathrm{Exp}(\mu_p)\ge \tau_p] = f_p \longrightarrow
    \tau_p=\frac{\ln(1/f_p)}{\mu_p}.
\]
Conditioned on exceeding the threshold $\tau_p$, memorylessness implies that the expected length of an interval is $\tau_p+1/\mu_p$. Since these intervals occur at rate $\mu_p f_p$ per unit time, page $p$ can contribute at most
\[
    \mu_p f_p \cdot \left(\tau_p + \frac{1}{\mu_p}\right) = f_p\left(1+\ln\frac{1}{f_p}\right)
\]
units of hole time per unit time whenever it is evicted on an $f_p$-fraction of its interarrival intervals.
Motivated by this observation, we define the (non-decreasing) \emph{coverage contribution} function
\[
\Phi(f):=
\begin{cases}
    f\left(1+\ln\frac{1}{f}\right), & f\in(0,1],\\
    0, & f=0,
\end{cases}
\]
mapping eviction probabilities $f\in[0,1]$ to hole coverage $\Phi(f)\in [0,1]$. 

We can now express the problem of maintaining $r$ holes on average from among pages in $S$ as that of finding fractions $f_p\in [0,1]$ that satisfy $\sum_{p\in S}\Phi(f_p)\ge r$, and minimize the induced cost rate $\sum_{p\in S}f_pq_p$. 
We obtain the following lower bound, formally proven in Appendix~\ref{appendix:optional_lb}.

\begin{lemma}[Optional Admission Lower Bound] \label{lemma:lb_opt}
The asymptotic cost rate of any optional admission policy for $\pnameit$ that serves pages in $S\subseteq[n]$ using a cache of size $k\ge 0$ is at least $\lopt(S,|S|-k)$, where
\[
\lopt(S,x):=
\min\left\{
    \sum_{p\in S} f_p q_p :
    f_p\in[0,1]\ \forall p\in S,\;\;
    \sum_{p\in S}\Phi(f_p)\ge x
\right\}.
\]
\end{lemma}

\subsection{Designing One-Hole Policies via Guarded Eviction}

Guided by the lower bound we just developed, we turn our attention to the special case of one-hole instances and prove a constant competitive ratio under perfect predictions.

\begin{lemma}[Optional Admission One-Hole Policy]\label{lemma:one-hole-opt}
    Under optional admission, given any instance of $\pnameit$ with pages $S$ and a cache of size $|S|-1$, and perfect per-request predictions, \Cref{alg:prefixguardedbelady} (Guarded Eviction) achieves a cost rate of at most $7\cdot \lopt(S,1)$.
\end{lemma}


We first develop and analyze a simple baseline policy that we call \textbf{Threshold Belady} (\Cref{alg:guardedbelady}), which uses predictions to guide its eviction decisions.
We then combine it with a threshold on intrinsic costs to obtain our final policy, \textbf{Guarded Eviction} (\Cref{alg:prefixguardedbelady}), which achieves the guarantee in \Cref{lemma:one-hole-opt}.

\subsubsection{The Threshold Belady Rule}
Fix any set $S$ of $m=|S|\ge 1$ pages and $k=m-1$ cache slots. 

We first provide some intuition. Suppose, as a thought experiment, that pages $p$ in $S$ have equal intrinsic costs $q_p:=\mu_p w_p=Q$. It is easy to see that the lower bound $\lopt(S,1)$ is then minimized by setting uniform fractions $f_p=f$ for all $p\in S$ such that $f\cdot (1+\log1/f)=1/m$. A straightforward computation yields $f=\Theta(1/(m\log m))$ and therefore we obtain that $\lopt(S,1)= \Theta(Q/\log m)$. Our goal, then, is to design a prediction-based policy that achieves cost rate $O(Q/\log m)$ over $S$. This target rate imposes a lower bound on the length of an interval over which it is worthwhile evicting some page $p$. In particular, amortizing the cost $w_p$ of evicting $p$ over an interval of length at least $w_p\log m/Q = \log m/\mu_p$ provides the desired rate. When such a page exists, it is safe to evict it. When no such page exists, we may need to take a default expensive action. The key observation is that because page arrivals are independent, the probability of this simultaneous failure decays exponentially in the number of pages, and therefore the extra cost imposed is small.

We synthesize these ideas into a simple policy, named $\name$ (\Cref{alg:guardedbelady}). The policy is parametrized by a (real) threshold $\theta>0$; pages qualify for eviction if their next arrival is at least $\theta/\mu_p$ units of time into the future. The following lemma establishes a bound on the performance of the policy. For instances with unequal intrinsic costs, we express the bound in terms of the maximum intrinsic cost of a page in $S$.

\begin{algorithm}[h]
\caption{$\name(S,\theta)$}
\label{alg:guardedbelady}
\begin{algorithmic}[1]
\Statex \textbf{Input:} A non-empty set of pages $S$ and a threshold \(\theta>0\).
\Whenever{a cache miss occurs at time $t$ on the unique current hole $p\in S$}
    \State Pay \(w_p\) and compute the qualifying set
    \[T:=\{p'\in S:\text{(next predicted $p'$ arrival)} \ge t+ \theta/\mu_{p'}\}.\]
    \State If $T=\emptyset$, keep $p$ as the hole; otherwise, choose any page in $T$ as the hole.
\EndWhenever
\end{algorithmic}
\end{algorithm}

\begin{lemma}[Performance of $\name$]\label{lemma:guarded-eviction}
    For any non-empty set of $n\ge 1$ pages $S$ and any threshold $\theta>0$, under perfect per-request predictions we have
    \[\rate(\name(S,\theta))\le \frac{1}{\theta}\max_{p\in S}q_p
    +
    (1-e^{-\theta})^n
    \sum_{p\in S}q_p.\]
    Choosing $\theta = 0.5\cdot\log n$ yields a rate of at most $3\cdot \max_{p\in S} q_p/\log n$.
\end{lemma}
\begin{proof}
We ignore the first miss, whose cost is at most $\max_{p\in S}w_p$ and therefore does not affect the asymptotic cost rate. We call a miss on page $p$ \emph{qualified} if $p$ belonged to the qualifying set $T$ when it was designated as the hole following the previous miss. Otherwise, we call the miss \emph{unqualified}: the qualifying set was empty and $p$ remained the hole. We bound the cost rates of these two types of misses separately.

\medskip
\textit{Qualified Misses.}
By definition, every qualified miss on page $p$ is preceded by a hole interval of length at least $\theta/\mu_p$. If this interval has length $\ell$, then
$$
    w_p
    \le
    \frac{w_p\mu_p}{\theta}\ell
     = \frac{q_p}{\theta}\ell \le
    \frac{\max_{p'\in S}q_{p'}}{\theta}\ell.
$$
The hole intervals corresponding to distinct qualified misses are disjoint, so their total length over a horizon of length $T$ is at most $T$. Dividing by $T$ and passing to the asymptotic rate, the total qualified cost rate is therefore at most
\begin{equation}\label{eq:ge1}
    \frac{\max_{p\in S}q_p}{\theta}.
\end{equation}

\textit{Unqualified Misses.}
Each unqualified miss on page $p$ can be charged to the preceding miss, which was an arrival of $p$ at which no page qualified. These charges are distinct, so the number of unqualified misses on $p$ is at most the number of such arrivals. We count all arrivals satisfying this condition, without conditioning on the algorithm's previous decisions, which may reveal information about future requests.

Consider an arrival time $t$ of page $p$. Immediately after this arrival, the next interarrival time of $p$ is a fresh $\mathrm{Exp}(\mu_p)$ random variable. Independently, for every other page $p'\neq p$, memorylessness implies that its residual time until the next arrival is distributed as $\mathrm{Exp}(\mu_{p'})$. Thus, each page fails to qualify with probability
$$
    \Pr\left[\mathrm{Exp}(\mu_{p'})<\frac{\theta}{\mu_{p'}}\right]
    =
    1-e^{-\theta}.
$$
These events are independent across pages, so the probability that no page qualifies is $(1-e^{-\theta})^n$. Since page $p$ generates arrivals at rate $\mu_p$, its arrivals at which no page qualifies occur at rate $\mu_p(1-e^{-\theta})^n$. Each unqualified miss on $p$ costs $w_p$, so the total unqualified cost rate is at most
\begin{equation}\label{eq:ge2}
    (1-e^{-\theta})^n
    \sum_{p\in S}w_p\mu_p=(1-e^{-\theta})^n
    \sum_{p\in S}q_p.
\end{equation}

Adding the bounds in \eqref{eq:ge1} and \eqref{eq:ge2} proves the first claim. Finally, setting $\theta=0.5\log n$ yields the second guarantee using the inequality
$(1-1/\sqrt{n})^n < 1/(n\log n)$, which holds for all $n\ge2$.
\end{proof}

\subsubsection{Guarding Against Expensive Evictions}
Because the cost rate of $\name$ depends on the intrinsic cost of the worst page in the instance $S$, the policy can be far from optimal over instances where different pages have very different intrinsic costs. A simple change fixes this issue: we keep the highest intrinsic cost pages permanently in cache and run the policy over the remaining pages. This change suffices to provide constant competitiveness against the lower bound established in \Cref{lemma:lb_opt}. \Cref{alg:prefixguardedbelady} formalizes the new policy $\namepr$. 



\begin{algorithm}[h]
\caption{$\namepr(S)$}
\label{alg:prefixguardedbelady}
\begin{algorithmic}[1]
\Statex \textbf{Input:} A non-empty set of pages $S$.
\State Let $n:=|S|$ and label the pages in $S$ by intrinsic costs $q_1\le q_2 \le \cdots \le q_n$.
\State Set $R_1:=q_1$ and compute $R_j:= 3\cdot (q_j/\log j)$ for each $j\ge 2$ and $j^\star \in \mathrm{argmin}_{j\in [n]}R_j$.
\State If $j^\star=1$ keep page $1$ permanently out of cache and pin pages $\{2,\ldots , n\}$.
\State Otherwise, run $\name(S',0.5\log j^\star)$ on the page prefix $S'=\{1,\ldots , j^\star\}$ and keep every page $p\in S\setminus S'$ permanently pinned in the cache.
\end{algorithmic}
\end{algorithm}

We will now prove \Cref{lemma:one-hole-opt}, showing that $\namepr$ is constant competitive with optional admission for the special case of $n=k+1$ and perfect per-request predictions.

\paragraph{Proof of \Cref{lemma:one-hole-opt}.}
Fix any set of $n\ge 1$ pages $S$ ordered as $q_1\le \cdots \le q_n$. By construction, $\namepr(S)$ returns a valid policy to serve the requests on $S$ that uses $|S|-1$ cache slots. Furthermore, by the choice of the prefix on which $\name$ is eventually run, coupled with the guarantee of \Cref{lemma:guarded-eviction} for our specific choice of threshold, the asymptotic cost rate of this policy is upper bounded by $R:=\min_{j\in [n]} R_j$. 

Next, fix any vector $f\in[0,1]^n$ satisfying $\sum_{p\in S}\Phi(f_p)\ge 1$ and let $L:=\sum_{p\in S}f_pq_p$. We will prove that $R\le 7L$ for any such vector; then, by taking the vector $f$ minimizing $\lopt(S,1)$, the proof immediately follows. Suppose, towards a contradiction, that this is not the case. We then have $q_1>7L$ and $ q_j>(7L/3)\cdot \log j$ for all $j\ge 2$. Moreover, since the pages are ordered by increasing intrinsic cost, $q_j\ge q_1>7L$ for every $j$. Thus, defining
$$
    a_j:=\max\left\{1,\frac{\log j}{3}\right\},
$$
we have $q_j>7La_j$ for every $j\in[n]$. Using the definition of $L$, this implies
\begin{equation}\label{eq:for_contr}
    \sum_{j=1}^na_jf_j<\frac17.
\end{equation}
The remaining argument is a simple calculation showing that any vector $f\in [0,1]^n$ satisfying \eqref{eq:for_contr} must necessarily satisfy $\sum_{j=1}^n\Phi(f_j) < 1$, giving us our contradiction and completing the proof. We present this calculation in Appendix \ref{appendix:computation} for completeness.

\subsection{A Constant-Loss Reduction to One-Hole Instances}
Finally, we explain how the constant-competitive guarantees of \Cref{lemma:one-hole-opt} for one-hole instances can be lifted into constant-competitive guarantees for the general case. Let $r:=n-k$ denote the number of holes we need to generate while serving the pages in $[n]$ and consider any partition of $[n]$ into non-empty, disjoint groups $\{S_j\}_{j=1}^r$. Then, one valid strategy is to separately implement a one-hole policy over each group; since we have $r$ groups the policy will generate the required number of holes and therefore be feasible; and by additivity of the cost objective, the total rate will be equal to the sum of the group rates. Therefore, if we use \Cref{alg:prefixguardedbelady} for the group-policies, the guarantee of \Cref{lemma:one-hole-opt} gives us a feasible policy for the global instance with rate at most 
\[7\cdot \sum_{j=1}^r\lopt(S_j,1).\]

The remaining algorithmic task is to (efficiently) derive a partition $\{S_j\}_{j=1}^r$ of the pages such that the total $\lopt(S_j,1)$ contribution is close to the global lower bound $\lopt([n],r)$. In fact, the lower bound directly guides us towards such a partition: from \Cref{lemma:lb_opt}, there exists some vector $f^\star\in[0,1]^n$ such that
$$
    \sum_{p\in[n]}\Phi(f_p^\star)\ge r
    \qquad\text{and}\qquad
    \lopt([n],r) \ge \sum_{p\in[n]} f_p^\star q_p.
$$
We can efficiently compute $f^\star$ through standard tools from convex-analysis. Furthermore, since the coverage contributions satisfy $\Phi(f^\star_p)\in [0,1]$ for all $p\in [n]$, the constraint $\sum_{p\in[n]}\Phi(f_p^\star)\ge r$ enables a standard packing argument that partitions the pages into disjoint sets $\{S_j\}_{j=1}^r$ with \[\sum_{p\in S_j}\Phi(f_p^\star)\ge 1/2\]  for all $j\in [r]$, by greedily packing them in decreasing $\Phi$-mass order until the requirement is satisfied.

Finally, by noting the identity $\Phi(6x)\ge 2\Phi(x)$ for all $x\in [0,1/6]$, we obtain that for each group $S_j$ the restriction $f'_p:=\min\{6f^\star_p,1\}$ is feasible for $\lopt(S_j,1)$, namely that $\sum_{p\in S_j}\Phi(f'_p)\ge 1$. Consequently, we have
\[\sum_{j=1}^r \lopt(S_j,1) \le \sum_{j=1}^r \sum_{p\in S_j}f'_pq_p = \sum_{p=1}^n f'_pq_p \le 6\sum_{p=1}^n f^\star_p q_p \le 6\rate(\opt).\]
For the sake of completeness, we formally state and analyze this process in Appendix \ref{app:opt-partition}.

\begin{lemma}[Optional Admission Partition] \label{lemma:partition_opt}
There exists an efficient routine, which we call $\greedypart$ (\Cref{alg:greedypartition}), that partitions any set of pages $[n]$ into $r$ non-empty subsets $\{S_j\}_{j=1}^r$ such that 
$\sum_{j=1}^r \lopt(S_j, 1) \le 6\cdot \lopt([n], r)$.
\end{lemma}

\subsection{Proof of \Cref{thm:const_opt}}

\begin{algorithm}[h!]
\caption{\textsc{OptionalAdmission}($S,r$)}
\label{alg:optional}
\begin{algorithmic}[1]
\Statex \textbf{Input:} A non-empty set of pages $S$ and a hole requirement $r\in\{0,1,\ldots, |S|\}$.
\State Let $\{S_j\}_{j=1}^r$ be the output of $\greedypart(S,r)$.
\State Run the one-hole policy $\namepr(S_j)$ for the pages of each component $S_j$.
\end{algorithmic}
\end{algorithm}

Guided by the above, we finally design a policy (\Cref{alg:optional}) for serving a set of pages $S$ using $k=|S|-r$ cache slots. The proof of \Cref{thm:const_opt} trivially follows by combining the lower bound of \Cref{lemma:lb_opt} with the one-hole policy of \Cref{lemma:one-hole-opt} and the partition of \Cref{lemma:partition_opt}.

\section{The Mandatory Admission Setting with Perfect Predictions}
\label{sec:mandatory}
Equipped with the intuition provided by the optional-admission setting, we now consider the mandatory-admission variant of $\pname$ under perfect per-request predictions. We prove \Cref{thm:perfect} from the introduction, which we re-state for the reader's convenience.

\perfectpredictiontheorem*

The main difficulty in extending the previous section is that a page can no longer remain permanently outside the cache. Whenever a missing page is requested, it must be admitted, forcing some other page to become the new hole. This affects each part of our optional-admission argument. We explain these complications and address them in three steps, following the same overall approach as in \Cref{sec:optional}: a lower bound on the offline optimum, a prediction-based algorithm for the special case of $n=k+1$, and a partition process that reduces the general setting to one-hole instances.

\subsection{A Strengthened Lower Bound on the Offline Optimal Rate}
\label{subsec:mandatory-lower-bounds}
The optional admission lower bound of~\Cref{lemma:lb_opt} only controls how much hole time each page can supply, and therefore remains valid under mandatory admission. However, it ignores the constraints on moving holes between pages, and under mandatory admission, such constraints matter. For example, with exactly $r$ zero-weight pages, we trivially have $\lopt(S,r)=0$. Yet, if all $r$ zero-weight pages are outside the cache, a request to any one of them forces the eviction of a positive-weight page. Thus, holes cannot remain confined to the free pages. 

To capture such interactions, we first consider a contracted instance that replaces a set of $r$ pages by a single zero-weight \emph{ghost page}, requested whenever any page in the set is requested. Keeping the same cache size produces a one-hole instance, allowing us to isolate the interaction between the ghost and the remaining pages.

\subsubsection*{Instance Contraction}
Fix a set of pages $S$, a cache of size $k<|S|$ and a strict subset $G\subset S$ of cardinality $|G|=|S|-k$. First, we set the weights of all pages in $G$ to zero, which can only decrease the offline optimal rate. We may then assume without loss that at most one page of $G$ is cached: a request to another page in $G$ can replace the cached representative at no cost, while keeping additional pages of $G$ only occupies space that could hold other pages. Thus, the pages in $G$ can be represented by a single zero-weight \emph{ghost page} $g$, requested whenever any page in $G$ is requested. By superposition, its requests form a Poisson process of rate $\mu_g:=\sum_{p\in G}\mu_p$. Thus, the resulting \textbf{contracted instance}, which we denote as $\con(S,G)$, has $k+1$ pages and the same $k$ cache slots, so it is a one-hole instance whose offline optimal rate is at most that of the original instance.

Conversely, a one-hole policy $\pol$ for $\con(S,G)$ can be simulated on the original pages: keep the same normal pages and, whenever the ghost is cached, represent it by the most recently requested page in $G$. Normal-page costs are unchanged, and restoring the original weights adds at most the cost of fetching every requested page in $G$. The resulting policy $\widehat{\pol}$ is therefore feasible for serving the pages in $S$ using $|S|-|G|$ cache slots and its rate satisfies
\[
    \rate(\widehat{\pol})
    \le \rate(\pol)+\sum_{p\in G}q_p.
\]
This allows us to lift guarantees from $n=k+1$ to general $n$, assuming that the set $G$ only contains pages of small intrinsic cost. We make use of this observation throughout this section.

\subsubsection*{Improving the Lower Bound}
Contraction allows us to strengthen the lower bound through a balance constraint on miss rates. Let $r:=|S|-k$ and fix any $G\subset S$ with $|G|=r$. In the corresponding contracted one-hole instance $\con(S,G)$, we may assume without loss that an optimal policy evicts the ghost whenever a normal page misses: the ghost's zero fetching cost allows us to postpone evicting a normal page until the next ghost request. Each ghost miss then evicts a normal page, whose next request brings it back and evicts the ghost again.

Ghost misses and normal-page misses therefore alternate, so their counts differ by at most one. Writing $f_p$ for the asymptotic miss fraction on page $p$, this translates to the balance constraint
\[
    \mu_g f_g=\sum_{p\in S\setminus G}\mu_p f_p.
\]
Thus, the relaxation can no longer assign the entire hole to the free ghost without also accounting for normal-page misses. Adding this constraint to the one-hole relaxation
from~\Cref{lemma:lb_opt}, applied to the contracted instance $\con(S,G)$, gives 
\[
\lghost(S,G):=\min\left\{
\sum_{p\in S\setminus G}q_p f_p:
f\in[0,1]^{(S\setminus G)\cup\{g\}} \;,\;
\Phi(f_g)+\sum_{p\in S\setminus G}\Phi(f_p)\ge1 \;,\;
\mu_g f_g=\sum_{p\in S\setminus G}\mu_p f_p
\right\}.
\]
where $\mu_g:=\sum_{p\in G}\mu_p$.

Every choice of $G$ with $|G|=r$ yields a valid lower bound $\lghost(S,G)$ on the offline optimal rate for $\con(S,G)$, and therefore the original instance. Combining these bounds with $\lopt(S,r)$ gives the following mandatory-admission counterpart of \Cref{lemma:lb_opt}, proved formally in Appendix~\ref{appendix:mandatory_lb}.

\begin{lemma}[Mandatory Admission Lower Bound]
\label{lemma:lb_man}
The asymptotic cost rate of any mandatory admission policy
for $\pnameit$ that serves pages in $S\subseteq[n]$ using a cache of size $k\ge 1$ is at least  $\lman(S,|S|-k)$, where
\[
    \lman(S,x):=
    \max\big\{
        \lopt(S,x),
        \max_{G\subset S:|G|=x}\lghost(S,G)
    \big\}.
\]
\end{lemma}

\subsection{A One-Hole Policy for Mandatory Admission}
Next, we design a one-hole policy for our setting. We begin with a mandatory admission variant of $\name$ (\Cref{alg:guardedbelady}), which serves as a subroutine in our final policy. When no qualifying page is available, $\name$ leaves the requested page outside the cache, which is not permitted under mandatory admission. Our variant instead designates the two pages with highest arrival rates as \emph{buffers} and alternates the hole between them until a non-buffer qualifies. The remaining pages, which we call \emph{anchors}, behave as in the optional admission variant and are only evicted when a sufficiently large inter-arrival window is identified. 

\begin{algorithm}[!htbp]
\caption{$\manpol(S,\theta)$}
\label{alg:mandguardedbelady}
\begin{algorithmic}[1]
\Statex \textbf{Input:} A set of at least two pages $S$ and a threshold $\theta>0$.
\State Let $B$ contain the two pages in $S$ with highest arrival rates $\mu_p$, and let $A:=S\setminus B$.
\Whenever{a cache miss occurs at time $t$ on the unique current hole $p\in S$}
    \State Pay $w_p$ and admit page $p$ into the cache.
    \If{$p\in B$}
        \State Compute the qualifying set $T:=\{a\in A:\text{(next predicted $a$ arrival)} \ge t+ \theta/\mu_{a}\}$.
        \State If $T\neq\emptyset$, evict an arbitrary page from $T$; otherwise, evict the other buffer in $B\setminus\{p\}$.
    \Else
        \State Evict the least recently requested buffer in $B$.
    \EndIf
\EndWhenever
\end{algorithmic}
\end{algorithm}

This extension is formally described in \Cref{alg:mandguardedbelady}. In essence, the unique hole is moved back and forth between the two buffers as they arrive, until a qualifying anchor is found, at which point it gets evicted. Choosing the two pages with the highest arrival rates as buffers gives the algorithm the most frequent opportunities to find a qualifying anchor. Once the evicted anchor is requested again, the least recently requested buffer is evicted, concluding the current cycle. Clearly, \Cref{alg:mandguardedbelady} uses $|S|-1$ cache slots and can be implemented efficiently under per-request predictions. Moreover, any requested page is fetched in the cache and it is therefore a feasible policy for the mandatory admission model. 

As before, by setting $\theta=c\log|S|$ for a suitable constant $c<1$, we can control the probability that no anchor is found. Furthermore, the higher the arrival rate $\mu_p$ of a page is, the smaller is its weight $w_p$ relative to its intrinsic cost $q_p=\mu_pw_p$ which allows us to control the buffer loss in our analysis. This becomes clear in the following counterpart of \Cref{lemma:guarded-eviction}, which is proven in Appendix \ref{appendix:mandatory-guarded-eviction} using similar arguments.

\begin{lemma}[Performance of $\manpol$]\label{lemma:guarded-eviction-man}
    For any set of at least two pages $S$ and any threshold $\theta>0$, under perfect per-request predictions we have
    \[\rate(\manpol(S,\theta))\le \left(\frac2\theta+(1-e^{-\theta})^{|S|-2}\right)\cdot \max_{p\in S}q_p.\]
\end{lemma}

Finally, as in the optional admission setting, rather than instantiating \Cref{alg:mandguardedbelady} over the entire set of pages, we only run it on a subset of them and keep the remaining permanently in the cache. We consider two types of alternatives for selecting this subset:
\begin{itemize}
    \item The first considers all prefixes $U$ of the pages (ordered by intrinsic cost), served via \Cref{alg:mandguardedbelady} with threshold $\theta=c\log|U|$ for some constant parameter\footnote{Throughout this section, we use $c=3/4$ but maintain the parametrization since \Cref{sec:predictions} uses $c=1/4$.} $c>0$. As before, this ensures that very expensive pages are never considered for eviction and controls the \Cref{lemma:guarded-eviction-man} dependency on the maximum intrinsic cost.

    \item  The second considers simple swap policies that alternate the hole between only two pages\footnote{Alternating between pages $p,p'$ results to each cycle incurring one miss on each page, at total cost $w_p+w_{p'}$ at expected duration $1/\mu_p+1/\mu_{p'}$. Its cost rate is therefore precisely $
    \frac{w_p+w_{p'}}{1/\mu_p+1/\mu_{p'}}
    =\frac{\mu_{p'}q_p+\mu_pq_{p'}}{\mu_p+\mu_{p'}}.
    $}; these policies exploit cheap evictions even when the evicted page does not meet the lateness threshold.  In particular, when applied to a contracted instance, swapping a zero-weight ghost with another page incurs at most one positive-cost miss per ghost request, allowing a low cost rate when ghost requests are rare even if the other page has large intrinsic cost.
\end{itemize}

\Cref{alg:mandatory-one-hole} (named $\ohp$) compares these candidates and runs the one with the smallest rate bound, which we denote by \(\cert(S,c)\). We state its performance in \Cref{lemma:one-hole-man}.

\begin{algorithm}[!htbp]
\caption{$\ohp(S,c)$}
\label{alg:mandatory-one-hole}
\begin{algorithmic}[1]
\Statex \textbf{Input:} A set of at least two pages $S$ and a threshold parameter $c>0$.

\Statex \textbf{Output:} A rate certificate $\cert$.

\State Set $m:=|S|$ and order pages by intrinsic cost
$q_1\le q_2\le\cdots\le q_m$.

\State Define the prefix-candidate sets $\mathcal{U}_{pref}:=
\bigl\{\{1,\ldots,j\}:3\le j\le m\bigr\}$ and assign certificate 
\[R(U):=\left(\frac{2}{c\log j}+(1-e^{-c\log j})^{j-2}\right)\cdot q_j\]
to each $U\in\mathcal{U}_{pref}$.
\State Define the swap-candidate sets $\mathcal{U}_{swap}:=
\bigl\{\{i,j\}:1\le i < j\le m\bigr\}$ and assign certificate 
\[R(U):=\frac{\mu_j q_i + \mu_i q_j}{\mu_i + \mu_j}\]
to each $U\in\mathcal{U}_{swap}$.

\State Choose 
\[
U^\star \in \mathrm{argmin}\{R(U):U\in\mathcal{U}_{pref}\cup\;\mathcal{U}_{swap}\}.
\]

\State Pin the pages in $S\setminus U^\star$ in the cache and serve $U^\star$ via 
$\manpol(U^\star,c\cdot\log|U^\star|)$.

\State \phantomsection\label{line:certificate-return}\textbf{Return} the certificate $\cert:=R(U^\star)$.
\end{algorithmic}
\end{algorithm}

\begin{corollary}[Performance of $\ohp$]\label{lemma:one-hole-man}
    For any set of at least two pages $S$ and any $c>0$, under perfect per-request predictions we have 
    \[\rate(\ohp(S,c))\le \cert(S,c)\]
    where $\cert(S,c)$ is the certificate returned in line \hyperref[line:certificate-return]{\ref*{line:certificate-return}} of \Cref{alg:mandatory-one-hole}.
\end{corollary}

\subsection{Partitioning the Pages Into Structured Components}

Under optional admission, we partitioned the pages into $r$ one-hole components and served each separately. Under mandatory admission, however, each such component must contain at least two pages: a singleton leaves no cache slot in which to admit its requests. Consequently, this decomposition is impossible when $2r>n$. We overcome this obstacle by allowing a residual component to accommodate several holes through the contraction introduced earlier in this section.

Our partition subroutine (\Cref{alg:mandatory-partition}) pools pages with small intrinsic costs into a set $G$, represented by a zero-weight ghost. Choosing cheap pages keeps the additional cost of simulating this contraction small. We allocate the remaining $h:=r-|G|$ holes to ordinary components, grouping pages according to their coverage in the optional-admission relaxation. Each group receives enough lower-bound mass to support an inexpensive one-hole policy. The leftover pages join $G$ in the residual component, which we serve through its contracted one-hole instance; if $G=\emptyset$, we simply pin these leftover pages.




\begin{algorithm}[!htbp]
\caption{$\manpar(S,r)$}
\label{alg:mandatory-partition}
\begin{algorithmic}[1]
\Statex \textbf{Input:} A set of pages $S$ and a hole requirement $r\in\{0,1,\ldots,|S|-1\}$.

\Statex \textbf{Output:} A collection of ordinary components
$\mathcal{T}:=\{S_j\}_{j=1}^h$ and a residual component $(S_0,G)$.

\If{at least $r$ pages in $S$ have $q_p=0$}
    \phantomsection\label{line:manpar-zero-test}
    \State \phantomsection\label{line:manpar-zero-returning}
    \textbf{Return} $\mathcal{T}=\emptyset$, $S_0=S$ and
    $G=\{\text{any $r$ pages in $S$ with $q_p=0$}\}$.
\EndIf

\State \phantomsection\label{line:manpar-coverage}
Set $u_p:=(1+q_p/\lambda)e^{-q_p/\lambda}$ for the unique
$\lambda>0$ satisfying
$\sum_{p\in S}(1+q_p/\lambda)e^{-q_p/\lambda}=r$.

\State \phantomsection\label{line:manpar-cheap}
Set $C:=\{p\in S:q_p\le6\lambda\}$.

\If{$|C|\ge r$}
    \phantomsection\label{line:manpar-cheap-test}
    \State \phantomsection\label{line:manpar-cheap-return}
    \textbf{Return} $\mathcal{T}=\emptyset$, $S_0=S$ and
    $G=\{\text{the $r$ pages in $S$ of smallest $q_p$}\}$.
\EndIf

\State \phantomsection\label{line:manpar-initialize}
Set $h:=r-|C|$ and $E:=S\setminus C$.

\For{$j=1,\ldots,h$}
    \phantomsection\label{line:manpar-components}
    \State \phantomsection\label{line:manpar-empty-component}
    Set $S_j:=\emptyset$.
    \While{$\sum_{p\in S_j}u_p<1/2$}
        \phantomsection\label{line:manpar-fill}
        \State \phantomsection\label{line:manpar-add-page}
        Choose a page $p\in E$ with smallest $q_p$ and set
        $S_j\gets S_j\cup\{p\}$ and $E\gets E\setminus\{p\}$.
    \EndWhile
\EndFor

\State \phantomsection\label{line:manpar-final-return}
\textbf{Return} $\mathcal{T}=\{S_j\}_{j=1}^h$,
$S_0=C\cup E$ and $G=C$.
\end{algorithmic}
\end{algorithm}

We state the following guarantee for \Cref{alg:mandatory-partition} below. We present a proof sketch here, demonstrating the key ideas required in the analysis. The formal proof is deferred to Appendix~\ref{app:man-partition}.

\begin{lemma}[$\manpar$ Guarantee]
\label{lem:mandatory-partition-feasibility}
For any set $S$ of at least two pages, and any hole requirement $r\in\{0,1,\ldots, |S|-1\}$, \Cref{alg:mandatory-partition} returns a partition $\{S_0,S_1,\ldots,S_h\}$ of $S$ into nonempty sets, together with a strict subset $G\subset S_0$, such that (i) $r=|G|+h$ and  (ii) $|S_j|\ge 2$ for all $j\in \{1,\ldots, h\}$. Furthermore, the partition components satisfy
\[
\sum_{j=1}^h \cert(S_j,\frac34) + \mathbbm{1}(G\neq \emptyset)\cdot\bigg(\sum_{p\in G}q_p +  \cert(\con(S_0,G),\frac34)\bigg) \le 15.75\cdot L_{\mathrm{man}}(S,r),
\]
where $\cert(S,c)$ is the certificate returned in line \hyperref[line:certificate-return]{\ref*{line:certificate-return}} of \Cref{alg:mandatory-one-hole}.
\end{lemma}
\begin{proofsketch}
The analysis depends on whether $h=0$ (no ordinary one-hole components are returned) or $h\ge1$. The latter case is very similar to optional admission, while the strengthened lower bound of \Cref{lemma:lb_man} is only required in the former case.

\paragraph{The case $h\ge 1$.}
Line~\hyperref[line:manpar-coverage]{\ref*{line:manpar-coverage}} computes the coverage values $u_p\in[0,1]$ from an optimal solution to $\lopt(S,r)$, with $\sum_{p\in S}u_p=r$. Each cheap page contributes at most one unit of coverage, so removing $C$ leaves at least $h=r-|C|$ units. Simple algebra shows that every expensive page has $u_p<1/2$, so each completed ordinary group contains at least two pages and consumes less than one unit of coverage. Consequently, all $h$ groups can be formed while leaving at least one expensive page in the residual component, establishing feasibility.

The coverage guarantees also control the one-hole certificates. Each ordinary component contains a constant amount of coverage, and the small overshoot when closing each group leaves a constant amount among the normal pages of the residual component as well. Since coverage decreases exponentially with $q_p/\lambda$, this forces each component to contain an inexpensive prefix: if every prefix certificate were much larger than $\lambda$, the ordered intrinsic costs would be too large to supply the required coverage. Thus, each one-hole certificate is $O(\lambda)$, and their total is $O(h\lambda)$, including the residual certificate.

Finally, the global relaxation pays for both these certificates and the ghost simulation. On cheap pages,
the optimal fractions $f_p^\star=e^{-q_p/\lambda}$ are bounded below by a positive constant, so the
relaxation pays a constant fraction of $\sum_{p\in G}q_p$. On expensive pages, it pays $\Theta(\lambda)$ per unit of coverage; since these pages supply at least $h$ units, it also pays $\Omega(h\lambda)$. Together, these observations bound the total certificate by $O(\lopt(S,r))$. A sharper joint accounting gives the stated constant.

\paragraph{The case $h=0$.}
If $h=0$, then the algorithm terminated at either Line~\hyperref[line:manpar-zero-returning]{\ref*{line:manpar-zero-returning}} or Line~\hyperref[line:manpar-cheap-return]{\ref*{line:manpar-cheap-return}}. In either case, $S_0=S$ and $|G|=r<|S|$, so feasibility is immediate. Both returns can be analyzed together. The simulation cost $\sum_{p\in G}q_p$ is controlled by $\lopt(S,r)$: it is zero in the first case, while in the second, the pages in $G$ are cheap and the preceding argument applies. Therefore, it only remains to control the contracted one-hole policy using the ghost lower bound.

Consider an optimal solution to the ghost lower bound $\lghost(S,G)$. Its miss fractions must satisfy
\[
    \Phi(f_g)+\sum_{p\in S\setminus G}\Phi(f_p)\ge1,
\]
reflecting the single hole in the contracted instance. Either the ghost term $\Phi(f_g)$ is at least $1/2$, or the normal-page terms sum to at least $1/2$. These two possibilities correspond to the two types of candidates considered by our one-hole policy (\Cref{alg:mandatory-one-hole}).

If $\Phi(f_g)\ge1/2$, then $f_g$ is bounded below by a positive constant. Let $D$ be the smallest certificate among swaps between the ghost and a normal page. By definition, $q_p\ge D(1+\mu_p/\mu_g)$ for every normal page $p$. The balance constraint therefore gives
\[
    L_{\mathrm{gst}}(S,G)
    =\sum_{p\in S\setminus G}q_pf_p
    \ge D\left(\sum_{p\in S\setminus G}f_p+f_g\right)
    \ge Df_g.
\]
Thus, one of the swap candidates has certificate $O(L_{\mathrm{gst}}(S,G))$. Intuitively, using the ghost to supply substantial hole time forces enough normal-page misses to pay for this swap.

Otherwise, the normal pages satisfy $\sum_{p\in S\setminus G}\Phi(f_p)\ge1/2$. Scaling their miss fractions by a constant makes them feasible for the optional-admission one-hole relaxation. Consequently,
\[
    \lopt(S\setminus G,1)
    =O\bigl(L_{\mathrm{gst}}(S,G)\bigr).
\]
The optional-admission prefix analysis therefore identifies an inexpensive prefix of normal pages. Adding the ghost to this prefix gives a feasible mandatory-admission candidate whose certificate increases by at most a constant factor. 

In either case, our policy includes a candidate with certificate $O(L_{\mathrm{gst}}(S,G))$. Combining this with the simulation bound completes the constant-factor comparison; the tighter quantitative analysis in the appendix yields the stated factor.
\end{proofsketch}

\subsection{Proof of \Cref{thm:perfect}}

\begin{algorithm}[!htbp]
\caption{$\manfinal(S,r,c)$}
\label{alg:mandatory-final}
\begin{algorithmic}[1]
\Statex \textbf{Input:} A set of at least two pages $S$, a hole requirement $r\in\{0,1,\ldots, |S|-1\}$ and a threshold parameter $c>0$.

\State Let $(\mathcal{T}, S_0, G)$ be the output of $\manpar(S,r)$.

\State For each ordinary component $S_j\in \mathcal{T}$, use $\ohp(S_j, c)$ to serve its pages using $|S_j|-1$ cache slots.

\State For the residual component $S_0$, if $G=\emptyset$ pin all pages in $S_0$ in the cache. Otherwise, use $\ohp(\con(S_0,G),c)$ on the contracted instance and translate it back into a policy for the original pages, using $|S_0|-|G|$ cache slots.

\end{algorithmic}
\end{algorithm}

Our final policy (\Cref{alg:mandatory-final}) first partitions $S$ and then applies the one-hole policy to each ordinary component and, when $G\neq\emptyset$, to the contracted residual instance $\con(S_0,G)$. The partition guarantees ensure that the resulting policy uses $|S|-r$ cache slots. With $c=3/4$, \Cref{lemma:one-hole-man} and the contraction simulation guarantee bound its total cost rate by the sum of certificates and the ghost simulation cost appearing on the left-hand side of \Cref{lem:mandatory-partition-feasibility}. This is at most $15.75\lman(S,r)$, which by \Cref{lemma:lb_man} is at most $15.75\rate(\opt)$, establishing the theorem.

\FloatBarrier

\section{Stochastic Weighted Caching with Prediction Errors}\label{sec:predictions}

Finally, we consider the mandatory-admission version of $\pname$ with erroneous per-request predictions. Recall that upon the time $t_p(i)$ when page $p$ arrives for the $i$-th time, the algorithm receives an arbitrary prediction $t'_p(i+1)$ for its next arrival that can depend on the entire realized request sequence as well as the algorithm and its internal randomness. The total prediction error is then defined as $\error:=\sum_{p\in [n]}\epsilon_p$, where
\[
    \epsilon_p
    := \limsup_{T\to \infty}\frac{\expect \big[
    \sum_{i: \; t_p(i)\le T}
    \left|t_p(i)-t'_p(i)\right|\big]}{T}.
\]
In this section, we give the formal proof of \Cref{thm:erroneous} from the introduction, which we re-state below for the reader's convenience:

\erroneouspredictiontheorem*

We first note that robustness follows from standard black-box techniques for combining online policies and achieving the best of their guarantees. \cite{blum97} give a randomized combiner with a multiplicative loss of $(1+\delta)$ for any fixed $\delta>0$ and a prediction-free fallback can be provided by the randomized $O(\log k)$-competitive weighted-caching algorithm of~\cite{bansal2012}, whose guarantees hold even under adversarial requests.

For a deterministic guarantee, we instead use the deterministic combiner of~\cite{antoniadis2020metric}, which incurs a loss of $9$ in the competitive ratio, together with a simple prediction-free fallback: contract the $r$ pages with the smallest intrinsic costs into a ghost, swap the ghost with the best uncontracted page, and pin all other pages. We present and analyze this policy, which we call $\robustname$ (\Cref{alg:robust}) in Appendix~\ref{app:robust}, showing that it is $O(\log k)$-competitive.

Therefore, it remains to establish consistency by providing an $O(1+\error)$-competitive policy. The algorithm we developed in the previous section achieves precisely that guarantee, under a slightly different parametrization: whenever we instantiate our Threshold Belady sub-routine over a page set $S$, we use a slightly smaller threshold $\theta=(1/4)\log|S|$ that will be crucial when addressing prediction errors, as we later show. Our final policy is presented below:

\begin{algorithm}[h]
\caption{$\errorname(S,r)$}
\label{alg:error-final}
\begin{algorithmic}[1]
\Statex \textbf{Input:} A set of at least two pages $S$ and a hole requirement $r\in\{0,1,\ldots, |S|-1\}$

\State Use the deterministic combiner of \cite{antoniadis2020metric}, instantiated over

\begin{enumerate}
    \item [(a)] The prediction-free policy $\robustname(S,|S|-r)$.
    \item [(b)] The prediction-based policy $\manfinal(S,r,1/4)$.
\end{enumerate}
\end{algorithmic}
\end{algorithm}

Note that this policy is not aware of the error, as desired. From our previous discussion, the proof of \Cref{thm:erroneous} follows immediately once we establish the following guarantee:

\begin{theorem}[Consistency]\label{thm:cons}
    For any set of at least two pages $S$ and any cache size $k\ge 1$, under erroneous per-request predictions of total error $\error$, we have 
    \[\rate(\manfinal(S,|S|-k,\frac{1}{4}))\le O(1+\error)\cdot \rate(\opt).\]
\end{theorem}

In fact, most of our analysis from the perfect-prediction case continues to carry over: both the partition sub-routine (\Cref{alg:mandatory-partition}) and the prefix-optimizing one-hole implementation (\Cref{alg:mandatory-one-hole}) are completely oblivious to the prediction model and only rely on the (still valid) lower bound of the offline optimum and optimizing over fixed rate certificates. Therefore, the only component that requires a different analysis and further insight is $\manpol$ (\Cref{alg:mandguardedbelady}) whose performance thus far has only been analyzed under perfect predictions.

Our main technical contribution in this section is the following analysis of \Cref{alg:mandguardedbelady} in the erroneous predictions model.

\begin{lemma}
\label{lemma:guarded-eviction-errors}
For any set of at least two pages $S$, under erroneous per-request predictions of total error $\error_S:=\sum_{p\in S}\epsilon_p$, we have
    \[\rate\big(\manpol(S,\frac{1}{4}\log|S|)\big)\le (16+40\error_S)\cdot \frac{\max_{p\in S}q_p}{\log|S|}.\]
\end{lemma}

Observe that up to a multiplicative factor of $O(1+\error_S)$, this guarantee is identical to the one we established in \Cref{lemma:guarded-eviction-man} under perfect predictions. Once this lemma is established, the proof of \Cref{thm:cons} follows immediately by propagating this extra error term throughout our analysis; we defer these details to Appendix \ref{app:consistency} and conclude this section with the proof of the lemma.

\subsection*{Proof of \Cref{lemma:guarded-eviction-errors}}\label{sec:error-analysis}

We use $B:=\{b_1,b_2\}$ to denote the buffer pages and $A:=S\setminus B$ for the anchors: by construction, $\mu_b\ge\mu_p$ for every $b\in B$ and $p\in A$. For notational convenience, we write $m:=|S|$, $Q:=\max_{p\in S}q_p$ and $\theta=(1/4)\log m$. We can assume without loss\footnote{Otherwise, even missing every request on every page gives rate $\sum_{p\in S}q_p \le 7Q \le (16Q)/\log m$ and the lemma trivially follows.} that $m\ge 8$ and therefore $\theta>1/2$. 

Recall that $\manpol(S,\theta)$ circulates its unique hole between the two buffer pages, which have the highest arrival rates in $S$, until an anchor qualifies: its predicted next arrival is at least $\theta/\mu_p$ time units away. The algorithm then evicts that anchor and waits until its next request before returning the hole to a buffer. Prediction errors can disrupt this process in two ways. A \emph{false positive} can cause an evicted anchor to return too soon, while a \emph{false negative} can hide an anchor with a sufficiently distant next arrival. We charge the resulting costs to overprediction and underprediction errors, respectively.

\paragraph{False Positives.}
Suppose a miss on buffer $b$ at time $s$ leads to the eviction of an anchor $p$, whose predicted and actual next arrival times are $t'$ and $t$, respectively. Since $p$ qualifies, we have $t'-s\ge\theta/\mu_p$. The buffer miss and the subsequent anchor miss therefore cost at most
\[
w_b+w_p = \frac{q_b}{\mu_b}+\frac{q_p}{\mu_p} \le Q\cdot\big(\frac{1}{\mu_b}+\frac{1}{\mu_p}\big) 
\le \frac{2Q}{\mu_p}
\le \frac{2Q}{\theta}\cdot\big((t-s)+(t'-t)\big).
\]
If $t'\le t$, the actual hole interval $[s,t)$ alone pays for this cost at rate $2Q/\theta$. Otherwise, we additionally charge the overprediction error $t'-t$ at the same rate. Note that the hole intervals of different anchor evictions are disjoint, and each prediction can cause at most one eviction of its page before the corresponding arrival replaces it. For a finite horizon $[0,T]$, we sum only over pairs with $t\le T$, so every charged error is included in the error budget through $T$. Any initial unpaired anchor miss and the buffer miss of the at most one unfinished pair have total cost at most $2\max_{p\in S}w_p$, which does not affect the asymptotic rate. Therefore, after summing over all the corresponding terms, and taking expectations and the limit superior as $T\to\infty$, we obtain that the total cost rate due to all anchor misses and all buffer misses at which some anchor qualified is at most
\begin{equation}\label{eq:false_positive}
    \frac{2Q}{\theta}\cdot (1+\error_S) \le \big(8 + 8\error_S\big)\cdot \frac{Q}{\log m}.
\end{equation}

\paragraph{False Negatives.}
The only misses not yet accounted for are buffer misses at which no anchor qualifies according to the predictions. Their analysis requires some care: conditioning on such a miss can reveal information about future requests, since both the current hole and the supplied predictions depend on the request sequence. We therefore consider \emph{all arrivals} of each buffer, without conditioning on whether they cause a miss or find a qualifying anchor. Therefore, for the purposes of our remaining analysis, each buffer $b\in B$ arrives at rate $\mu_b$ and causes a miss at cost $w_b$ if no anchor qualifies according to the predictions.

Fix an arrival of buffer $b$ at time $s$. We call it \textbf{typical} if at least $2\theta$ anchors have true next arrival times $t_p\ge s+2\theta/\mu_p$, and the interval $[s-\theta/\mu_b,s+\theta/\mu_b]$ contains at most $8\theta$ arrivals of $b$, including $s$. Otherwise (i.e. if either of the two conditions doesn't hold) we call it \textbf{exceptional}.

At a typical buffer miss with no qualifying anchor, the first condition provides many anchors whose actual gaps are at least twice the qualification threshold. Their predictions must therefore be early by at least $\theta/\mu_p$. This margin lets us distribute the cost of the search among these anchors and over a short interval. The second condition limits how many buffer arrivals can contribute charges over that interval, thereby controlling the total charge density. 

Meanwhile, lowering the threshold to $\theta=(1/4)\log m$ compared to the perfect-prediction case ensures that anchors with actual gaps exceeding $2\theta/\mu_p$ are sufficiently abundant. A concentration bound controls the probability of having too few such anchors, while a separate Poisson tail bound controls excessive buffer arrivals. Together, these bounds control the frequency of exceptional buffer arrivals and hence their contribution to the cost rate.

We formalize this intuition below; we bound the cost rate generated by typical and exceptional buffer misses at which no anchor qualified in \eqref{eq:typical} and \eqref{eq:exceptional} respectively, and then add these contributions together with \eqref{eq:false_positive} in order to conclude the proof.

\paragraph{Typical Buffer Misses.}
For accounting purposes, we charge each typical buffer miss with no qualifying anchor to the erroneous anchor predictions that caused the search to fail. Fix a buffer $b$ and write $h_b:=\theta/\mu_b$. At such a miss at time $s$, at least $2\theta$ anchors have true next arrival times $t_p\ge s+2\theta/\mu_p$. However, their current predictions satisfy $t'_p<s+\theta/\mu_p$, since none qualified. Each of these anchor predictions therefore underestimates its next arrival time by more than $\theta/\mu_p$. Assign a charge of $w_b/(2\theta)$ to each such prediction; together, these charges cover the miss cost $w_b$.

To control repeated charges to the same anchor prediction, we further spread each assigned charge evenly over $[s,s+h_b]$, at a density of $w_b/(2\theta h_b)$ per unit of time. We bound the total charge to a fixed prediction by controlling where these charging intervals lie and how many can overlap.

In particular, fix an anchor prediction $t'_p$ whose actual target arrival is at time $t_p$. Every search at time $s$ that charges this prediction satisfies
\[
    t'_p-\frac{\theta}{\mu_p}<s
    \qquad\text{and}\qquad
    s+h_b\le t_p-\frac{\theta}{\mu_p},
\]
where the second inequality follows from $t_p\ge s+2\theta/\mu_p$ and $h_b\le\theta/\mu_p$. Consequently, every charging interval $[s,s+h_b]$ is contained in
\[
    \left[t'_p-\frac{\theta}{\mu_p},
          t_p-\frac{\theta}{\mu_p}\right].
\]
All charge to this prediction therefore lies within an interval of length $t_p-t'_p$, equal to its underprediction error.

To bound overlaps, fix any time $x$ receiving charge. Every contributing search must occur in $[x-h_b,x]$. Choose one such search at time $s$. Its typicality window $[s-h_b,s+h_b]$ contains $[x-h_b,x]$, so the second typicality condition bounds the number of contributing searches by $8\theta$. Since each contributes density $w_b/(2\theta h_b)$, the combined charge density is at most
\[
    \frac{8\theta w_b}{2\theta h_b}
    =\frac{4q_b}{\theta}.
\]
Integrating over the interval supporting all charges, the total charge to this prediction from buffer $b$ is at most
\[
    \frac{4q_b}{\theta}(t_p-t'_p).
\]

At a finite horizon $T$, each anchor $p$ has at most one prediction issued by $T$ whose actual target arrival is after $T$. Let $a_p(T)$ be the time since $p$'s last request, or $T$ if none has occurred. Charges from searches up to $T$ to this prediction are supported on $[T-a_p(T),T+h_b]$, so their total from buffer $b$ is at most $(4q_b/\theta)(a_p(T)+h_b)$. Since $\mathbb{E}[a_p(T)]\le1/\mu_p$, these outstanding charges have total expectation $O(1)$ after summing over anchors and buffers and do not affect the asymptotic rate.

Finally, summing over anchor predictions and both buffers, and using $q_{b_1}+q_{b_2}\le2Q$, the total cost rate of these misses is at most
\begin{equation}\label{eq:typical}
    \frac{8Q}{\theta}\error_S
    =32\error_S\cdot\frac{Q}{\log m}.
\end{equation}

\paragraph{Exceptional Buffer Misses.}
At an arrival of buffer $b$, the anchors' residual arrival times are independent exponentials. Each anchor therefore has its next arrival at least $2\theta/\mu_p$ time units away with probability $e^{-2\theta}$. The number of such anchors is binomial with mean
\[
    (m-2)e^{-2\theta}
    =\frac{m-2}{\sqrt m}
    \ge \log m=4\theta,
\]
where the inequality holds for $m\ge8$. Since the required number $2\theta$ is at most half this mean, a multiplicative Chernoff bound gives
\[
    \Pr\{\text{first condition fails}\}
    \le \exp\left(-\frac{(m-2)e^{-2\theta}}8\right)
    \le m^{-1/8}
    =e^{-\theta/2}.
\]

The number of other arrivals of $b$ within distance $\theta/\mu_b$ has a Poisson distribution with mean at most $2\theta$, allowing for truncation at time zero. Since $\theta\ge1/2$, failure of the second condition requires more than $8\theta-1\ge6\theta$ other arrivals. A Poisson tail bound therefore gives
\[
    \Pr\{\text{second condition fails}\}
    \le e^{-[6\log3-4]\theta}
    \le e^{-2\theta}
    =m^{-1/2}.
\]
By a union bound, a buffer arrival is exceptional with probability at most $e^{-\theta/2}+e^{-2\theta}$. Counting every exceptional arrival as a miss and applying the Poisson arrival-counting formula, the total cost rate of exceptional buffer misses is at most
\begin{equation}\label{eq:exceptional}
    (q_{b_1}+q_{b_2})
    \bigl(e^{-\theta/2}+e^{-2\theta}\bigr)
    \le \frac{2Q}{\theta}
    =8\cdot\frac{Q}{\log m},
\end{equation}
where we used
$\theta(e^{-\theta/2}+e^{-2\theta})
\le 2/e+1/(2e)<1$.

\section{Declaration of AI Usage} 
Initial iterations of the algorithms and analysis presented in this work were developed through interactions with ChatGPT 5.6 Sol. The authors directed the exploration, evaluated proposed constructions, and made key decisions about which approaches to pursue and how to develop and present the arguments. ChatGPT 6 Astra was also used to optimize the final constant in Theorem 1.1 and help with the presentation of the paper. The authors have verified the correctness of the algorithms and proofs and assume full responsibility for the contents.

\newpage
\appendix
\section*{Appendix}
\begin{appendices}
\addtocontents{toc}{\protect\setcounter{tocdepth}{1}}

\section{Discrete vs Continuous Arrival Models}\label{app:discussion}
Our continuous-time Poisson model is equivalent, at the level of request order, to the standard discrete-time model in which requests are drawn independently and in successive time-steps from a fixed distribution. We briefly formalize this connection. Let
\[
    \mu := \sum_{p\in[n]} \mu_p
    \qquad\text{and}\qquad
    d_p := \frac{\mu_p}{\mu}.
\]
Let $\mathcal{D}$ be the distribution over pages that assigns probability $d_p$ to page $p$. 
Consider the superposition of the independent Poisson request processes of all pages. By standard properties of Poisson processes, the resulting request stream is itself a Poisson process of rate $\mu$. Moreover, independently at each arrival, the requested page is $p$ with probability $d_p$. 
Thus, if we ignore the arrival timestamps, the sequence of requested pages is precisely an i.i.d.\ sequence drawn from the distribution $\mathcal{D}$. Conversely, given any discrete-time i.i.d.\ request distribution $\mathcal{D}$, attaching the requests to the arrivals of an independent rate-$\mu$ Poisson process yields, by Poisson thinning, independent page processes of rates 
$\mu_p=\mu d_p$. 
Hence \textbf{every discrete-time i.i.d.\ instance corresponds to a continuous-time SWC instance, and vice versa.}

For coupled executions, the two cost normalizations also give the same competitive guarantees. In particular, let
\[
    \rate_{\mathrm{cont}}(\pol)
    :=
    \limsup_{T\to\infty}
    \frac{\expect[\cost_{\pol}(T)]}{T}
\]
be the (asymptotically) average cost of policy $\pol$ per unit of time, as defined in \Cref{sec:preliminaries}. Letting $T_m$ denote the time of the $m$-th arrival in the continuous model, the discrete analogue of the rate is simply
\[
    \rate_{\mathrm{disc}}(\pol)
    =
    \limsup_{m\to\infty}
    \frac{\expect[\cost_{\pol}(T_m)]}{m},
\]
and it captures the (asymptotically) average cost of policy $\pol$ per request. We claim that
\[
    \rate_{\mathrm{cont}}(\pol)
    =
    \mu\cdot \rate_{\mathrm{disc}}(\pol).
\]
Indeed, let $N(T)\sim\mathrm{Poisson}(\mu T)$ denote the total number of requests by time $T$, and set $m=\lfloor\mu T\rfloor$. Writing $w_{\max}:=\max_{p\in[n]}w_p$, the two horizons $T$ and $T_m$ contain exactly $|N(T)-m|$ different requests, and hence
\[
    \big|\cost_{\pol}(T)-\cost_{\pol}(T_m)\big|
    \le
    w_{\max}|N(T)-m|.
\]
Since $\expect[|N(T)-m|]=O(\sqrt{T})$, taking expectations gives
\[
    \expect[\cost_{\pol}(T)]
    =
    \expect[\cost_{\pol}(T_{\lfloor\mu T\rfloor})]+o(T).
\]
Dividing by $T$ and taking the limit superior, while observing that
$\lfloor\mu T\rfloor/T\to\mu$, proves the claimed identity. The same argument applies to the offline optimum, and therefore, when its rate is positive,
\[
    \frac{\rate_{\mathrm{cont}}(\pol)}
         {\rate_{\mathrm{cont}}(\opt)}
    =
    \frac{\rate_{\mathrm{disc}}(\pol)}
         {\rate_{\mathrm{disc}}(\opt)}.
\]
Thus, changing the cost normalization for a coupled execution does not change its competitive guarantee.

Finally,the corresponding per-request prediction models differ. In continuous time, a perfect prediction reveals the exact time of the next request to the current page, while in discrete time a perfect PRP typically reveals its next request index and therefore additionally reveals the number of intervening requests. Our perfect-prediction upper bounds transfer through a randomized simulation: given perfect predicted request indices, one may independently generate a rate-$\mu$ Poisson clock and use the resulting arrival times as synthetic continuous-time predictions. The induced request process has exactly the distribution considered above, and hence our continuous-time policies can be simulated without any loss in competitive ratio. The resulting discrete-time policies are randomized, even when the continuous-time policies are deterministic. The simulation also preserves the error-independent robustness guarantee under arbitrary predictions.

For erroneous predictions, our error-dependent guarantee transfers provided that the predicted indices are independent of the synthetic Poisson clock, conditional on the request sequence. The prediction adversary may still know the entire request sequence. Let $p_j$ denote the page requested at index $j$, and let $\hat j$ be the prediction supplied for this request at the preceding request to the same page. Define the accumulated index error and its asymptotic rate by
\[
    A_p(m):=\sum_{\substack{j\le m\\p_j=p}}|j-\hat j|,
    \qquad
    \mathcal E_{\mathrm{disc}}
    :=\sum_{p\in[n]}
    \limsup_{m\to\infty}\frac{\expect[A_p(m)]}{m}.
\]
Under this independence assumption, the Poisson spacing identities give
\[
    \expect\!\left[
        \sum_{\substack{j\le N(T)\\p_j=p}}
        |T_j-T_{\hat j}|
    \right]
    \le \frac{1}{\mu}\expect[A_p(N(T))].
\]
Indeed, for fixed indices $j,\hat j$, the expected timing error restricted to $T_j\le T$ equals
$|j-\hat j|\Pr(N(T)\ge j)/\mu$ when $\hat j\ge j$; when $\hat j<j$, exchangeability of the first $j$ Poisson spacings gives the same expression with $\Pr(N(T)\ge j+1)$ instead. Since the clock is independent of the indices and $\expect[N(T)]=\mu T$, dividing by $T$ and taking limit superiors shows that the induced continuous-time error rate is at most $\mathcal E_{\mathrm{disc}}$. Consequently, our error-dependent competitive guarantee transfers in terms of discrete index error under this prediction model.
\section{A Constant Lower Bound for SWC with Perfect PRP}\label{app:lb}

In this chapter of the appendix we prove for the continuous-time PRP model, which we restate for the reader's convenience.

\lowerbound*

Lower bounds in competitive analysis typically exploit an online algorithm's uncertainty about future requests through an adversarially chosen input sequence. In $\pname$, however, requests are generated by fixed stochastic processes, so we cannot choose their realization. Instead, we identify events that occur with positive frequency and on which hindsight offers an expected advantage, even with perfect per-request predictions. Specifically, we construct a family of disjoint time intervals within the request horizon satisfying:
\begin{enumerate}
    \item[(P1)] The intervals occur at an expected rate of at least $\lambda>0$.

    \item[(P2)] For every online policy $\pol$ with perfect per-request predictions, conditioned on an interval's occurrence, a clairvoyant replacement saves at least $\Delta>0$ in expected cost while matching $\pol$'s cache at both endpoints.
\end{enumerate}
Matching the boundary caches allows us to replace $\pol$'s behavior inside these intervals without changing its behavior outside them. The comparison policy continues to simulate $\pol$ virtually, including its internal state and random choices, and follows this virtual execution outside the replacement intervals. Since the intervals are disjoint, the savings add, yielding
\[
    \rate(\pol)-\rate(\opt)\ge \lambda\Delta.
\]
The key is therefore to identify a local decision whose preferred action depends on future requests that perfect per-request predictions have not yet revealed. As we show, this separation already holds for $n=3$ pages and a cache of size $k=2$. Specifically, consider pages $\{a,b,c\}$ requested by independent Poisson processes with uniform rates $\mu_a=\mu_b=\mu_c=1$. The two light pages are identical, with $w_a=w_b=1$, while the heavy page has weight $w_c=5/2$. 

\subsection*{Defining the Replacement Intervals}

For the occurrence construction and counting arguments, extend the independent request processes to stationary Poisson processes on $\mathbb R$ to avoid boundary effects. This extension is used only in the analysis: the policy still starts at time $0$, and replacements are made only on intervals contained in $[0,T]$.

Fix a request to a light page $p\in\{a,b\}$ at time $t$, and let $q=\{a,b\}\setminus \{p\}$ denote the other light page. We call time $t$ a \emph{raw occurrence} if:
\begin{enumerate}
    \item The request sequence strictly before $t$ ends with $q\mid c^+\mid q^+\mid p^*$.

    \item The request sequence strictly after $t$ begins with $q^+\mid p\mid \{p,q\}^*\mid c^+\mid p^+\mid q$.
\end{enumerate}
Here, vertical bars separate segments, $x^+$ denotes one or more consecutive requests to $x$, $x^*$ denotes zero or more such requests, and $\{p,q\}^*$ denotes an arbitrary, possibly empty, sequence of requests to the light pages. 

Let $s_t$ be the time of the initial displayed $q$-request in the prefix, namely the last request preceding the displayed $c$-run. Let $e_t$ be the time of the terminal displayed $q$-request in the suffix. Associate with the raw occurrence the interval $I_t:=(s_t,e_t]$. We retain the occurrence if $I_t$ does not overlap the interval of any earlier raw occurrence. This exclusion applies even if that earlier occurrence was itself discarded. The retained intervals are therefore pairwise disjoint.

We first establish (P1). Requests to the two light pages arrive at total rate $2$, and each is a candidate for the distinguished $p$-request of an occurrence. The conditions defining raw occurrences and the rule for retaining them reduce this rate by only a constant factor.

\begin{lemma}\label{lem:occurrence_rate}
Let $N(T)$ denote the number of retained occurrence intervals contained in $[0,T]$. Then
\[
    \lim_{T\to\infty}\frac{\expect[N(T)]}{T}
    =\frac{23}{2304}.
\]
\end{lemma}

We next establish (P2). At the distinguished request to $p$, perfect per-request predictions reveal the next request times
$t_q<t_p<t_c$, but not whether an additional request to $q$ will occur between $t_p$ and $t_c$ (i.e. whether $q\in \{p,q\}^*$). The two possibilities favor different choices about retaining the heavy page $c$. A clairvoyant policy can exploit this uncertainty while matching the online policy's cache at both endpoints of the interval.

\begin{lemma}\label{lem:occurrence_loss}
Fix any online policy $\pol$ with perfect per-request predictions. For an occurrence interval $I=(s,u]$ with $s\ge0$, let $C^*_{\pol}(I)$ be the minimum cost of a feasible clairvoyant policy on $I$ that matches $\pol$'s cache immediately after $s$ and immediately after $u$. Then
\[
    \expect\!\left[
        \cost_{\pol}(I)-C^*_{\pol}(I)
        \,\middle|\,
        I\text{ is retained},\ s\ge0
    \right]
    \ge \frac18.
\]
\end{lemma}

For each realization, a policy attaining $C^*_{\pol}(I)$ provides
a feasible replacement for $\pol$ inside $I$. Since the retained
intervals are disjoint and each replacement preserves the boundary
caches, their savings add. For boundary accounting, use the savings truncated at $1/2$: the proof of \Cref{lem:occurrence_loss} gives the same conditional expectation bound of $1/8$ for these truncated savings, since it uses only the lower bound in \eqref{eq:local_loss}. Sum first over retained occurrences whose distinguished request lies in $[0,T]$ and whose interval starts at or after $0$. The proof of \Cref{lem:occurrence_rate} shows that their expected number is $(23/2304)T+O(1)$, and that only $O(1)$ in expectation end after $T$. Discarding the latter loses at most $1/2$ per interval. Consequently,
\[
    \expect[\cost_{\pol}(T)]-\expect[\cost_{\opt}(T)]
    \ge \frac{23}{2304}\cdot\frac18\,T-O(1).
\]
Dividing by $T$ and taking limit superiors yields
\[
    \rate(\pol)-\rate(\opt)
    \ge \frac{23}{2304}\cdot\frac18
    =\frac{23}{18432}>0.001.
\]
Finally, consider the policy that keeps $c$ permanently cached and uses the remaining slot for the most recently requested light page. Light-page requests arrive at total rate $2$, and each differs from the preceding light-page request with probability $1/2$. This policy therefore has cost rate $1$, implying $\rate(\opt)\le1$. Consequently, we have $\rate(\pol)\ge 1.001\rate(\opt)$ which proves \Cref{thm:lowerbound}. The rest of this section is dedicated towards proving the two lemmas.

\subsection{Proof of \Cref{lem:occurrence_rate}}
We first compute the rate of raw occurrences and then account for the exclusion of overlapping intervals. The superposition of the three request processes is a rate-$3$ Poisson process whose request labels are independent and uniform over $\{a,b,c\}$. Fix a request to a light page $p$ at time $t$, and let $q$ be the other light page. Looking backward from $t$, the first label different from $p$ is $q$ with probability
\[
    \sum_{r\ge0}\left(\frac13\right)^{r+1}=\frac12.
\]
This sum allows arbitrarily many $p$-requests before the distinguished request. Similarly, the first label preceding that $q$-run must be $c$, and the first label preceding that $c$-run must be $q$, each with conditional probability $1/2$. Thus the required prefix has probability
\[
    \Pr\!\left(q\mid c^+\mid q^+\mid p^*
        \text{ immediately before }t\right)=\frac18.
\]
Independently of the past, the future must begin with $q^+\mid p$, which has probability
\[
    \sum_{m\ge1}\left(\frac13\right)^{m+1}=\frac16.
\]
After this first future $p$-request, the segment $\{p,q\}^*$ places no restriction on the requests before the first subsequent
$c$-request, which occurs almost surely. After that specified
$c$-request, the remaining suffix must be $c^*\mid p^+\mid q$. Its probability is $1/4$: the first subsequent label different
from $c$ must be $p$, and the first label subsequently different from $p$ must be $q$, each with probability $1/2$. Therefore, a light-page request is a raw occurrence with probability
\[
    \frac18\cdot\frac16\cdot\frac14=\frac1{192}.
\]
Since light-page requests arrive at total rate $2$, raw occurrences have rate
\begin{equation}\label{eq:raw_occurrence_rate}
    \rho_{\mathrm{raw}}=2\cdot \frac1{192}=\frac1{96}.
\end{equation}

We next determine the fraction of raw occurrences that are retained. There is no multiple counting within a $p$-run: only its last request can satisfy the required future prefix $q^+\mid p$. Index the $c$-runs in chronological order. Every raw occurrence
uses two consecutive $c$-runs, and there is at most one raw occurrence with a given initial $c$-run. Indeed, its orientation is determined by the cheap-page runs surrounding that $c$-run, and its distinguished request must be the last request in the second cheap-page run following it.

Only raw occurrences with consecutive initial $c$-runs can
overlap. To see this, an occurrence using $c$-runs $C_i,C_{i+1}$
ends at the first $q$-request in its suffix
$C_{i+1}\mid p^+\mid q$, before the next $c$-run. Any occurrence
starting at $C_{i+2}$ or later begins immediately after the last
cheap-page request preceding its initial $c$-run, which cannot
precede this endpoint.

Now fix a current raw occurrence with prefix
$q\mid C_i\mid q^+\mid p^*$, followed by its distinguished
$p$-request. Let $L$ be the number of cheap-page runs between
$C_{i-1}$ and $C_i$. This segment ends in $q$. Looking backward,
each preceding distinct label is either $c$, which terminates
the segment, or the other light page, which extends it, with
equal probability. Consequently,
\[
    \Pr(L=\ell\mid\text{current prefix})=2^{-\ell},
    \qquad \ell\ge1.
\]

An overlapping predecessor exists exactly when $L$ is even and
at least $4$, and the cheap-page run immediately preceding
$C_{i-1}$ is $p$. In that case, the compressed request pattern is
\[
    p\mid C_{i-1}\mid
    \underbrace{p\mid q\mid p\mid q\mid\cdots\mid p\mid q}
        _{L\text{ cheap-page runs}}
    \mid C_i\mid q\mid p.
\]
The predecessor has the opposite orientation: its distinguished
request is the last $q$ in the second cheap-page run after
$C_{i-1}$, and its interval ends at the first $p$ in the current
$p$-run. Conversely, any overlapping predecessor must have this
form. The condition $L\ge4$ supplies its required initial
alternation, while even parity ensures that the segment starts
in $p$ and ends in $q$.

Conditional on $L$, the required label preceding $C_{i-1}$ is
$p$ with probability $1/2$. Hence the exclusion probability is
\[
    \Pr(\text{overlapping predecessor}\mid
        \text{current raw occurrence})
    =\frac12\sum_{j=2}^{\infty}2^{-2j}
    =\frac1{24}.
\]
This test depends only on the history through the current
distinguished request, so conditioning on the current occurrence's
future requirements does not change the calculation. Moreover,
the exclusion rule checks all earlier raw occurrences, regardless
of whether they were retained. Thus
\begin{equation}\label{eq:retained_occurrence_rate}
    \rho_{\mathrm{retained}}
    =\rho_{\mathrm{raw}}\left(1-\frac1{24}\right)
    =\frac1{96}\cdot\frac{23}{24}
    =\frac{23}{2304}.
\end{equation}

Finally, let $\widehat N(T)$ count retained occurrences whose
distinguished request lies in $[0,T]$. Stationarity and
\eqref{eq:retained_occurrence_rate} give
\[
    \expect[\widehat N(T)]=\frac{23}{2304}T.
\]
The associated intervals have finite expected length: the required
runs have geometric lengths, and the waiting times for the
specified pages have finite means. Therefore, the expected number
of such intervals crossing either endpoint of $[0,T]$ is $O(1)$. Therefore, it follows that
\[
    \expect[N(T)]=\frac{23}{2304}T+O(1),
\]
and dividing by $T$ proves the lemma.

\subsection{Proof of \Cref{lem:occurrence_loss}}

Fix a candidate request to a light page $p$ at time $t$, let $q$ be the other light page, and condition on this occurrence being
retained and its interval starting at or after time zero. Write $I=(s,u]$ for its associated interval, and let $t_q<t_p<t_c$ be the next request times of $q,p,c$ after time $t$, respectively.
Define
\[
    \kappa:=\frac52\,
    \mathbf{1}\{c\notin K_{\pol}(s)\},
\]
where $K_{\pol}(\tau)$ denotes $\pol$'s cache immediately after the request at time $\tau$. Thus $\kappa$ is the compulsory cost of the first $c$-request in the interval. By mandatory admission, both $K_{\pol}(s)$ and $K_{\pol}(u)$ contain page $q$. Next, we partition the possible continuations into the two events
\[
    S:=\{\text{no $q$-request occurs in }(t_p,t_c)\},
    \qquad
    L:=\{\text{at least one $q$-request occurs in }(t_p,t_c)\}.
\]
The event $L$ allows arbitrarily many additional requests to either
light page. We first show that
\begin{equation}\label{eq:local_loss}
    \cost_{\pol}(I)-C^*_{\pol}(I)
    \ge
    \begin{cases}
        \frac12\mathbf{1}_S,
            &\text{if }c\notin K_{\pol}(t),\\[2mm]
        \frac12\mathbf{1}_L,
            &\text{if }c\in K_{\pol}(t).
    \end{cases}
\end{equation}
As usual, we may assume that $\pol$ is lazy, maintaining a full
cache and changing it only when serving a miss.

We begin by constructing feasible comparison policies that match
both boundary caches. In either case, the comparison policy
serves the initial $c$-run while retaining $q$, paying $\kappa$
if $c$ is initially absent. It therefore holds $\{q,c\}$ before
the first $p$-request.

On $S$, the comparison policy retains $c$ until the final
$q$-request. Ignoring consecutive repetitions, its light-page
misses are on the sequence $p,q,p,q$, giving total cost
$\kappa+4$. Immediately before the final $q$-request, its cache
is $\{p,c\}$, so it can choose its final eviction to match either
possible cache $K_{\pol}(u)$. Hence
\begin{equation}\label{eq:short_comparison}
    C^*_{\pol}(I)\le \kappa+4
    \qquad\text{on }S.
\end{equation}

On $L$, the comparison policy evicts $c$ at the first $p$-request
and retains both light pages until $t_c$. All intervening
light-page requests are then hits. At $t_c$, it fetches $c$ and
evicts $q$; the suffix $p^+\mid q$ incurs one further miss.
Again, its final eviction can match either possible cache
$K_{\pol}(u)$. Thus
\begin{equation}\label{eq:long_comparison}
    C^*_{\pol}(I)
    \le \kappa+1+\frac52+1
    =\kappa+\frac92
    \qquad\text{on }L.
\end{equation}

Suppose first that $c\notin K_{\pol}(t)$. The policy must have
incurred at least one light-page miss between the initial
$c$-run and $t$. It must also pay $5/2$ to fetch $c$ at $t_c$,
and at least one light-page miss in the suffix $p^+\mid q$,
because admitting $c$ displaces one of the two light pages.
Consequently,
\[
    \cost_{\pol}(I)\ge\kappa+\frac92.
\]
Together with \eqref{eq:short_comparison}, this gives a loss of
at least $1/2$ on $S$.

Now suppose that $c\in K_{\pol}(t)$ and $L$ occurs. If $\pol$
evicts $c$ before $t_c$, it has already incurred at least two
light-page misses: the first $p$-request and the $q$-request at
$t_q$. It must subsequently pay $5/2$ to reload $c$ and at least
one light-page miss in the suffix. Its cost is therefore at least
$\kappa+11/2$. Alternatively, if $\pol$ retains $c$ through $t_c$, it incurs
at least four light-page misses on the alternation $p,q,p,q$.
If there are further light-page alternations before $t_c$, these
already force a fifth miss. Otherwise, its cache at $t_c$ is
$\{q,c\}$, and the following $p$-request forces a fifth miss.
Thus, in either case,
\[
    \cost_{\pol}(I)\ge\kappa+5.
\]
Combining this with \eqref{eq:long_comparison} gives a loss of
at least $1/2$ on $L$. Since
$\cost_{\pol}(I)\ge C^*_{\pol}(I)$ always, this proves
\eqref{eq:local_loss}.

It remains to quantify the uncertainty between $S$ and $L$.
Let $\mathcal F_t^{\mathrm{PRP}}$ denote the policy's complete
information immediately after serving the distinguished request,
including its predictions and internal random choices through
that decision. In particular, $t_q,t_p,t_c$ and $K_{\pol}(t)$
are measurable with respect to this information. Set
\[
    H:=t_c-t_p.
\]
Although the next $q$-request at $t_q$ is known, requests to $q$
after $t_q$ have not yet been revealed. Conditional on the
information available at $t$, these subsequent requests form
a fresh rate-one Poisson process. Furthermore, the overlap
exclusion depends only on the past, and the required suffix
after $t_c$ is independent of the requests in $(t_p,t_c)$.
Therefore, on a retained occurrence,
\[
    \Pr(S\mid\mathcal F_t^{\mathrm{PRP}},\,I\text{ is retained}
        ,\ s\ge0)
        =e^{-H},
    \qquad
    \Pr(L\mid\mathcal F_t^{\mathrm{PRP}},\,I\text{ is retained}
        ,\ s\ge0)
        =1-e^{-H}.
\]
The two continuations thus remain possible with the same exact
predictions at the distinguished request. Applying
\eqref{eq:local_loss}, regardless of the policy's decision, gives
\[
    \expect\!\left[
        \cost_{\pol}(I)-C^*_{\pol}(I)
        \,\middle|\,
        \mathcal F_t^{\mathrm{PRP}},\,I\text{ is retained}
        ,\ s\ge0
    \right]
    \ge \frac12\min\{e^{-H},1-e^{-H}\}.
\]

Finally, before conditioning on the policy's information, the
three next-request residuals are independent rate-one exponentials. Conditional on their ordering $t_q<t_p<t_c$,
memorylessness gives $H\sim\mathrm{Exp}(1)$. Neither the
past-based exclusion nor the independent suffix condition
changes this distribution. The additional condition $s\ge0$ also depends only on the past and does not change this distribution. Hence $e^{-H}$ is uniform on $[0,1]$ when averaging over retained occurrences with $s\ge0$. Taking expectations,
we obtain
\[
    \expect\!\left[
        \cost_{\pol}(I)-C^*_{\pol}(I)
        \,\middle|\,
        I\text{ is retained},\ s\ge0
    \right]
    \ge \frac12\int_0^1\min\{v,1-v\}\,dv
    =\frac18,
\]
as claimed.

\section{Supplemental Material for \Cref{sec:optional} (Optional Admission)}
In this chapter of the appendix, we present all omitted proofs and details from \Cref{sec:optional}.

\subsection{Formal Proof of the Lower Bound (\Cref{lemma:lb_opt})}\label{appendix:optional_lb}
We give a more formal proof of our lower bound using the key intuition that was sketched in the main body. Fix any set of $S$ of $n:=|S|$ pages and let $r:=n-k\in \{0,1,\ldots , n\}$ be the number of holes we need to maintain. For any $T>0$, let $M_p(T)$ denote the (random) number of misses incurred by the offline optimal policy ($\opt$) on page $p\in S$ during horizon $[0,T]$, and let
$$
H_p(T)
:=
\int_0^T
\mathbf{1}\{p\text{ is not in $\opt$'s cache at time }t\}\,dt
$$
denote the total amount of time during $[0,T]$ for which page $p$ is a hole. Since $\opt$ must maintain at least $r$ holes at any time $t$, feasibility implies that 
\begin{equation}\label{eq:opt-lb-hole-feasible}
   \sum_{p\in S} H_p(T) \ge rT \quad\text{with probability $1$.}
\end{equation}

We first relate the hole duration of a page to the number of misses on that page. Fix a page $p$ and a parameter $\tau_p\ge 0$ to be determined later. Consider any interarrival interval of page $p$, of some length $L$. Since page $p$ is requested at both endpoints of the interval, the page can become a hole at most once, since $\opt$ is assumed lazy without loss. If it does become a hole, the corresponding hole interval has length at most
$$
L \le \max\{\tau_p,L\} = \tau_p+(L-\tau_p)^+.
$$
Therefore, summing over the interarrival intervals of page $p$ contained in $[0,T]$ and adding a boundary term, we have
\begin{equation}\label{eq:opt-lb-duration-bound}
    H_p(T)
    \le
    \tau_p M_p(T)
    +
    \sum_{\ell}(L_{p,\ell}-\tau_p)^+
    +B_p(T)
\end{equation}
where the sum ranges over \textit{all} the interarrival intervals of page $p$ contained in $[0,T]$ to remove any dependencies on whether page $p$ was actually evicted during an interval and its length. Here $B_p(T)$ is the sum of the time before the first request and the time after the last request to $p$ within $[0,T]$, taking both terms to be $T$ if no request occurs. Each term has expectation at most $1/\mu_p$, so $\mathbb{E}[B_p(T)]\le 2/\mu_p$.

Since the interarrival times of page $p$ are exponential with rate $\mu_p$, namely $L_{p,\ell}\sim\operatorname{Exp}(\mu_p)$, we furthermore obtain
$$
\mathbb{E}\left[(L_{p,\ell}-\tau_p)^+\right]
=
\int_{\tau_p}^\infty \Pr(L_{p,\ell}>t)\,dt
=
\int_{\tau_p}^\infty e^{-\mu_p t}\,dt
=
\frac{e^{-\mu_p\tau_p}}{\mu_p}.
$$
By independence of pages $p$'s arrivals, a standard renewal-reward argument therefore implies that the expected sum over complete interarrival intervals satisfies
$$
\mathbb{E}\left[\sum_{\ell}(L_{p,\ell}-\tau_p)^+\right]
=
T e^{-\mu_p\tau_p}+o(T).
$$
Choose a sequence $T_j\to\infty$ along which $\mathbb{E}[M_p(T_j)]/T_j$ converges for every $p\in S$. Such a sequence exists because these quantities lie in $[0,\mu_p]$ and $S$ is finite.
Taking expectations in~\eqref{eq:opt-lb-duration-bound}, dividing by $T_j$, and letting $T_j\to\infty$ therefore gives
\begin{equation}\label{eq:opt-lb-hole-rate}
\limsup_{T_j\to\infty}\frac{\mathbb{E}[H_p(T_j)]}{T_j}
\le
\tau_p x_p+e^{-\mu_p\tau_p},
\end{equation}
where
$$
x_p
:=
\lim_{T_j\to\infty}\frac{\mathbb{E}[M_p(T_j)]}{T_j}
$$
is the limiting normalized expected miss count on page $p$ along this subsequence.

Now let $f_p:=x_p/\mu_p$ be the normalized version of this quantity; since the number of misses on a page is at most its number of requests, we have that $f_p\in[0,1]$ for each $p\in S$. Finally, selecting the (bound-minimizing) thresholds
\[
\tau_p := \frac{1}{\mu_p}\log\frac{1}{f_p},
\]
with the case $f_p=0$ handled by letting $\tau_p\to\infty$ in~\eqref{eq:opt-lb-hole-rate}, and combining with~\eqref{eq:opt-lb-hole-feasible} and~\eqref{eq:opt-lb-hole-rate}, we have that 
\[
\sum_{p\in S}\Phi(f_p)\ge r.
\]

It therefore remains to argue that $\rate(\opt)\ge \sum_{p\in S}\mu_p w_p f_p = \sum_{p\in S}x_p w_p$.
Indeed, along the same subsequence, we have
\[
\sum_{p\in S}x_p w_p
=
\lim_{T_j\to\infty}
\frac{\sum_{p\in S}w_p\,\mathbb{E}[M_p(T_j)]}{T_j}
=
\lim_{T_j\to\infty}
\frac{\mathbb{E}[\operatorname{cost}_{\opt}(T_j)]}{T_j}
\le
\rate(\opt),
\]
and the claim follows.

\subsection{Omitted Calculation in the Proof of \Cref{lemma:one-hole-opt}}\label{appendix:computation}

From \eqref{eq:for_contr}, it suffices to show that for any vector $f\in [0,1]^n$ for which
\[
\sum_{j=1}^na_jf_j<\frac17 \quad , \quad a_j:= \max\left\{1,\frac{\log j}{3}\right\}
\]
we have
\[
\sum_{j=1}^n \Phi(f_j) = \sum_{j=1}^n f_j \big(1 + \log\frac{1}{f_j}\big) < 1.
\]

We first use the elementary inequality $f\left(1+\log\frac1f\right)
\le sf+e^{-s}$ which is valid for every $f\in(0,1]$ and $s>0$. Applying it to each $f_j$ with $s_j=5.25a_j$ gives

$$
    \sum_{j=1}^n\Phi(f_j) \le \sum_{j=1}^n 5.25a_j f_j +
    \sum_{j=1}^n e^{-5.25a_j}
    <
    \frac{5.25}{7}
    +
    \sum_{j=1}^n e^{-5.25a_j}.
$$

It remains to bound the second term. Since $e^3\in(20,21)$, we have $a_j=1$ for $j\le20$, while $a_j=\log(j)/3$ for $j\ge21$. Therefore,

$$
    \sum_{j=1}^n e^{-5.25a_j}
    \le
    20e^{-5.25}
    +
    \sum_{j=21}^{\infty}j^{-7/4}.
$$
Using the integral bound for the latter sum,

$$
    20e^{-5.25}
    +
    \sum_{j=21}^{\infty}j^{-7/4}
    \le
    20e^{-5.25}
    +
    \int_{20}^{\infty}x^{-7/4}\,dx
    =
    20e^{-5.25}
    +
    \frac{4}{3}20^{-3/4}
    <\frac14.
$$
Consequently,

$$
    \sum_{j=1}^n\Phi(f_j)
    <
    \frac{5.25}{7}+\frac14
    =
    1,
$$
contradicting the assumed feasibility of $f$ and concluding the proof.

\subsection{The $\greedypart$ Sub-Routine (Proof of \Cref{lemma:partition_opt})}\label{app:opt-partition}

The following subroutine groups pages greedily according to their coverage in the global lower-bound relaxation.

\begin{algorithm}[h]
\caption{$\greedypart(S,r)$}
\label{alg:greedypartition}
\begin{algorithmic}[1]
\Statex \textbf{Input:} A non-empty set of pages $S$ and a hole requirement $r\in\{0,1,\ldots, |S|\}$.
\Statex \textbf{Output:} A partition $\mathcal{T}:=\{S_j\}_{j=1}^r$ of $S$ into $r$ non-empty sets.

\If{$r=|S|$}
\State \textbf{Return} each page in $S$ as a singleton set.
\EndIf

\If{at least $r$ pages in $S$ have $q_p=0$}
\State Choose distinct pages $z_1,\ldots,z_r$ with $q_{z_j}=0$.
\State Set $S_j:=\{z_j\}$ for $j\in[r-1]$, $S_r:=S\setminus\{z_1,\ldots,z_{r-1}\}$ and \textbf{return} $\mathcal{T}=\{S_j\}_{j=1}^r$.
\EndIf

\State Set $u_p:=(1+q_p/\lambda)e^{-q_p/\lambda}$ for the unique $\lambda>0$ satisfying $\sum_{p\in S}(1+q_p/\lambda)e^{-q_p/\lambda}=r$.

\State Set $E:=S$.

\For{$j=1,\ldots,r-1$}
\State Set $S_j:=\emptyset$.
\While{$\sum_{p\in S_j}u_p<1/2$}
\State Choose a page $p\in E$ with smallest $q_p$.
\State Set $S_j\gets S_j\cup\{p\}$ and $E\gets E\setminus\{p\}$.
\EndWhile
\EndFor

\State Set $S_r:=E$ and \textbf{return} $\mathcal{T}=\{S_j\}_{j=1}^r$.
\end{algorithmic}
\end{algorithm}

We now prove \Cref{lemma:partition_opt}. Recall that 
\[
\lopt(S,x):=
\min\left\{
    \sum_{p\in S} f_p q_p :
    f_p\in[0,1]\ \forall p\in S,\;\;
    \sum_{p\in S}\Phi(f_p)\ge x
\right\}.
\]
 If $r=|S|$, feasibility of the relaxation requires $f_p=1$ for every page, so
\[
\lopt(S,r)=\sum_{p\in S}q_p
=\sum_{j=1}^r\lopt(S_j,1).
\]
If at least $r$ pages have zero intrinsic cost, then $\lopt(S,r)=0$: assigning $f_p=1$ to $r$ such pages and zero to all others gives a feasible solution of cost zero. Each returned component also contains a zero-cost page, so $\lopt(S_j,1)=0$ for every $j\in[r]$. Thus, both early returns satisfy the lemma.

In all other cases, the number of zero-cost pages is strictly smaller than $|S|$. The KKT conditions for the relaxation give an optimal solution of the form
\[
f_p^\star=e^{-q_p/\lambda}\qquad , 
\qquad
\sum_{p\in S}\Phi(f_p^\star)=r.
\]
The sum on the left is continuous and strictly increasing in $\lambda>0$, with limits equal to the number of zero-cost pages and $|S|$ as $\lambda\to0$ and $\lambda\to\infty$, respectively. Hence, the multiplier exists, is unique, and can be found by one-dimensional binary search. The algorithm therefore uses the coverage values
\[
u_p=\Phi(f_p^\star),
\qquad
\sum_{p\in S}u_p=r.
\]

We next show that the algorithm returns a partition into $r$ non-empty sets, each with coverage at least $1/2$. Selecting pages by non-decreasing intrinsic cost is equivalent to selecting them by non-increasing coverage. Every page with coverage at least $1/2$ therefore forms a singleton group. Thereafter, every selected page has coverage below $1/2$, so a group closed by the loop has total coverage less than $1$: immediately before its last page is added, its coverage is below $1/2$, and that page contributes less than $1/2$. Thus, every completed group consumes at most one unit of coverage.

Inductively, before constructing group $j\le r-1$, the available pages have coverage at least $r-(j-1)$, so the loop can always reach its target. After constructing these $r-1$ groups, the remaining pages have coverage at least $1$ and form the final group $S_r$. Every page belongs to exactly one group, and
\[
\sum_{p\in S_j}\Phi(f_p^\star)\ge\frac12
\qquad\text{for every }j\in[r].
\]

Finally, define $f'_p:=\min\{6f_p^\star,1\}$. We claim that its restriction to each component $S_j$ is feasible for $\lopt(S_j,1)$. If some $p\in S_j$ satisfies $f_p^\star\ge1/6$, then $f'_p=1$, which alone supplies the required coverage. Otherwise, $f'_p=6f_p^\star$ for every $p\in S_j$. For $0<x\le1/6$, we have
\[
\frac{\Phi(6x)}{\Phi(x)}
=6\frac{1+\log(1/x)-\log6}{1+\log(1/x)}
\ge\frac{6}{1+\log6}>2.
\]
Consequently,
\[
\sum_{p\in S_j}\Phi(f'_p)
\ge2\sum_{p\in S_j}\Phi(f_p^\star)
\ge1,
\]
establishing feasibility. Summing the costs of these feasible solutions over the partition gives
\[
\sum_{j=1}^r\lopt(S_j,1)
\le\sum_{p\in S}q_pf'_p
\le6\sum_{p\in S}q_pf_p^\star
=6\lopt(S,r),
\]
which completes the proof.

\section{Supplemental Material for \Cref{sec:mandatory} (Mandatory Admission)}
In this chapter of the appendix, we present all omitted proofs and details from \Cref{sec:mandatory}.

\subsection{Formal Proof of the Lower Bound
(\Cref{lemma:lb_man})}\label{appendix:mandatory_lb}

Let $r:=|S|-k$. The case $r=0$ is trivial, so assume $1\le r\le |S|-1$. Every mandatory-admission policy is also feasible under optional admission, so \Cref{lemma:lb_opt} gives $\rate(\opt)\ge\lopt(S,r)$. It remains to show that $\rate(\opt)\ge\lghost(S,G)$ for every $G\subset S$ with $|G|=r$.

\paragraph{Reducing to the Contracted Instance.}
Fix such a set $G$ and write $N:=S\setminus G$, so $|N|=k$. Consider $\con(S,G)$ under its natural coupling with the original instance: each request to a page in $G$ becomes a request to the ghost $g$, while normal-page requests remain unchanged.

For each horizon $T>0$, we construct a policy $\pol$ for the contracted instance that simulates $\opt$ on $[0,T]$. The constructed policy may depend on $T$. Let $C$ and $C'$ denote their respective caches. We maintain the invariant
\[
    C\cap N\subseteq C'\cap N,
\]
so every normal page cached by $\opt$ is also cached by $\pol$. Up to a bounded initialization cost, $\pol$ starts with all pages in $N$ cached and processes requests as follows:
\begin{itemize}
    \item On a miss on a normal page $p\in N$, admit $p$ and evict $g$, leaving every normal page cached.

    \item On a miss on $g$, admit $g$ and evict any page in $N\setminus C$, where $C$ is the cache of $\opt$ after serving the corresponding original request. Such a page exists because mandatory admission requires $C$ to contain the requested page in $G$, while $|C|\le k=|N|$.
\end{itemize}
On a hit, $\pol$ leaves its cache unchanged. These updates preserve the invariant: $\opt$ can add a normal page only when that page is requested, and $\pol$ also has it cached after serving the request. Consequently, every normal-page miss of $\pol$ is also a miss of $\opt$, while ghost misses cost zero. Thus,
\[
    \operatorname{cost}_{\pol}(T)
    \le
    \operatorname{cost}_{\opt}(T)+O(1),
\]
where the initialization term is bounded independently of $T$.

\paragraph{Obtaining the Balance Constraint.}
By construction, every normal-page miss of $\pol$ evicts the ghost, and every ghost miss evicts a normal page. Since the contracted instance has exactly one hole, ghost misses and normal-page misses alternate. Writing $M_p(T)$ for the number of misses on page $p$ by time $T$ in the execution constructed for that horizon, we therefore have
\[
    \left|M_g(T)-\sum_{p\in N}M_p(T)\right|\le1.
\]

Choose a sequence $T_j\to\infty$ along which $\expect[M_p(T_j)]/(\mu_pT_j)$ converges for every $p\in N\cup\{g\}$, and denote these limits by $f_p$. Such a sequence exists because all these quantities lie in $[0,1]$ and there are finitely many pages. The preceding inequality gives
\[
    \mu_g f_g=\sum_{p\in N}\mu_p f_p.
\]
Applying the hole-time argument from Appendix~\ref{appendix:optional_lb} along the same subsequence also yields
\[
    \Phi(f_g)+\sum_{p\in N}\Phi(f_p)\ge1.
\]
Thus, $f$ is feasible for the relaxation defining $\lghost(S,G)$, and hence
\[
    \lghost(S,G)
    \le\sum_{p\in N}q_pf_p
    =
    \lim_{T_j\to\infty}
    \frac{\sum_{p\in N}w_p\,\expect[M_p(T_j)]}{T_j}
    \le\rate(\opt),
\]
where the last inequality follows from the finite-horizon cost comparison above. Since $G$ was arbitrary, combining this inequality
with $\rate(\opt)\ge\lopt(S,r)$ proves
\Cref{lemma:lb_man}.

\subsection{Analysis of $\manpol$
(\Cref{lemma:guarded-eviction-man})}
\label{appendix:mandatory-guarded-eviction}
We follow exactly the same steps as in the proof of \Cref{lemma:guarded-eviction}, adjusting some arguments to account for the introduction of buffers and anchors. 
Write $n:=|S|$, $Q:=\max_{p\in S}q_p$, and let $B:=\{b_1,b_2\}$ and $A:=S\setminus B$ denote the buffers and anchors, respectively.

We call a buffer miss \emph{qualified} if it finds a qualifying anchor, and \emph{unqualified} otherwise. Each qualified buffer miss is paired with the subsequent miss on the anchor it evicts. These pairs account for all anchor misses. We separately bound the cost rates of these pairs and the unqualified buffer misses. Initialization and the at most one unfinished pair contribute only a bounded cost and do not affect the asymptotic rate.

\medskip
\textit{Qualified Misses.}
Suppose a qualified miss on buffer $b$ evicts an anchor
$p$. Under perfect predictions, the resulting hole
interval has length $\ell\ge\theta/\mu_p$.
Since buffers have the highest arrival rates 
$(\mu_b\ge\mu_p)$, the two misses cost at most
\[
    w_b+w_p
    =\frac{q_b}{\mu_b}+\frac{q_p}{\mu_p}
    \le\frac{2Q}{\mu_p}
    \le\frac{2Q}{\theta}\ell.
\]
The hole intervals corresponding to distinct qualified
pairs are disjoint, so their total length over a horizon
of length $T$ is at most $T$. Dividing by $T$ and passing
to the asymptotic rate, their total cost rate is therefore
at most
\begin{equation}\label{eq:mge1}
    \frac{2Q}{\theta}.
\end{equation}

\medskip
\textit{Unqualified Misses.}
These are buffer misses at which no anchor qualifies.
As in the optional-admission proof, we count a larger
collection of events to avoid conditioning on the
algorithm's decisions, which may reveal information
about future requests.

Consider a virtual one-slot cache serving only the
two buffers and always retaining the most recently
requested buffer. Whenever the algorithm's hole is
a buffer, its other buffer coincides with the page
in the virtual cache. This holds after an unqualified
buffer miss, which evicts the other buffer, and after
an anchor miss, which evicts the least recently
requested buffer. Hits preserve this property.
Consequently, every unqualified miss is also a virtual
miss on the same page at a time when no anchor qualifies.

The virtual miss process depends only on buffer requests
and is therefore independent of the anchor arrival
processes. At any virtual miss time $t$, the anchors'
residual arrival times are independent exponentials,
so
\[
    \Pr[\text{no anchor qualifies at }t]
    =
    \prod_{p\in A}
    \Pr\left[\mathrm{Exp}(\mu_p)<\frac{\theta}{\mu_p}\right]
    =
    (1-e^{-\theta})^{n-2}.
\]
The virtual cache implements a two-page swap policy,
whose total cost rate is
\[
    \frac{w_{b_1}+w_{b_2}}
         {1/\mu_{b_1}+1/\mu_{b_2}}
    =
    \frac{\mu_{b_2}q_{b_1}+\mu_{b_1}q_{b_2}}
         {\mu_{b_1}+\mu_{b_2}}
    \le Q.
\]
By independence, restricting to virtual misses at which
no anchor qualifies multiplies this cost rate by
$(1-e^{-\theta})^{n-2}$. Since every unqualified miss
is included among these events, the total unqualified
cost rate is at most
\begin{equation}\label{eq:mge2}
    (1-e^{-\theta})^{n-2}Q.
\end{equation}
Adding the bounds in \eqref{eq:mge1} and \eqref{eq:mge2}
proves the lemma.

\subsection{Analysis of \Cref{alg:mandatory-partition} (Proof of \Cref{lem:mandatory-partition-feasibility})}\label{app:man-partition}

If $r=0$, the algorithm returns $S_0=S$, $G=\emptyset$ and no ordinary components, so both feasibility and the certificate bound are immediate. Henceforth, assume $r\ge1$. Write $Q_G:=\sum_{p\in G}q_p$, and let 
\[
\cert_0:= \begin{cases}
    \cert(\con(S_0,G),3/4) &\text{ if } G\neq \emptyset \\
    0 &\text{ otherwise } \\
\end{cases}
\]
The quantity we must bound is therefore $\sum_{j=1}^h\cert(S_j,3/4)+Q_G+\cert_0$.

Whenever the algorithm reaches Line~\hyperref[line:manpar-coverage]{\ref*{line:manpar-coverage}}, fewer than $r$ pages have zero intrinsic cost. Writing $\psi(x):=(1+x)e^{-x}$, the function $\sum_{p\in S}\psi(q_p/\lambda)$ is continuous and strictly increasing in $\lambda>0$, with limits equal to the number of zero-cost pages and $|S|$, respectively. Thus, the required $\lambda$ exists and is unique. The optimality conditions for $\lopt(S,r)$ give
\begin{equation}\label{eq:manpart-global-solution}
    f_p^\star=e^{-q_p/\lambda},
    \qquad
    u_p=\left(1+\frac{q_p}{\lambda}\right)f_p^\star,
    \qquad
    \sum_{p\in S}u_p=r,
    \qquad
    \lopt(S,r)=\sum_{p\in S}q_pf_p^\star.
\end{equation}

\subsection*{The case $h\ge1$.}
In this case, $G=C$ and $h=r-|C|$. Set $\delta:=\psi(6)=7e^{-6}<1/2$. Every expensive page satisfies $0<u_p<\delta$, while removing the cheap pages leaves coverage
\[
    \sum_{p\in S\setminus C}u_p
    =r-\sum_{p\in C}u_p
    \ge r-|C|
    =h.
\]
Each completed ordinary component has coverage in $[1/2,1/2+\delta)$: its coverage is below $1/2$ before the final page is added, and that page contributes less than $\delta$. Consequently, before constructing component $j\le h$, the remaining coverage is at least $h-(j-1)(1/2+\delta)\ge1$. The loop can therefore construct all $h$ components. Each contains at least two pages, since no expensive page alone contributes coverage $1/2$.

After these components have been constructed, the remaining expensive pages satisfy
\begin{equation}\label{eq:manpart-component-coverage}
    \sum_{p\in S_j}u_p\ge\frac12 \quad (j\in[h]),
    \qquad
    \sum_{p\in S_0\setminus G}u_p
    >h\left(\frac12-\delta\right)>0.
\end{equation}
Thus, $S_0\setminus G$ is nonempty, so $S_0$ is nonempty and $G\subsetneq S_0$. Each selected page is removed from the available set and assigned to exactly one ordinary component; the remaining pages form $S_0$. Hence, the returned sets partition $S$, and $|G|+h=|C|+(r-|C|)=r$.

\paragraph{Bounding the component certificates.}
We first record a convenient consequence of the prefix and swap candidates. For $i\ge2$, define
\[
    a_i:=\max\left\{1,\frac{5}{16}\log i\right\}.
\]
For any set $U$ ordered locally as $q_1\le\cdots\le q_m$, we claim that
\begin{equation}\label{eq:manpart-prefix-bound}
    \cert(U,3/4)\le\frac{q_i}{a_i}
    \qquad\text{for every }i\in\{2,\ldots,m\}.
\end{equation}
For $i\le24$, we have $a_i=1$, and the swap between the two cheapest pages has certificate at most $q_2\le q_i$. For $i\ge25$, the prefix of length $i$ has certificate
\[
    \left(\frac{8}{3\log i}+(1-i^{-3/4})^{i-2}\right)q_i
    \le \frac{16}{5\log i}\,q_i
    =\frac{q_i}{a_i}.
\]
Indeed, $\log x\cdot e^{-(x-2)x^{-3/4}}$ is decreasing for $x\ge25$ and is less than $8/15$ at $x=25$, so
\[
    \log i\left(\frac{8}{3\log i}+(1-i^{-3/4})^{i-2}\right)
    \le\frac83+\log i\cdot e^{-(i-2)i^{-3/4}}
    <\frac83+\frac8{15}
    =\frac{16}{5}.
\]

The coverage guarantees now force an inexpensive prefix in every component. The numerical estimate needed for this argument is
\begin{equation}\label{eq:manpart-coverage-series}
\begin{aligned}
    \delta+\sum_{i=2}^{\infty}\psi\left(\frac{27}{4}a_i\right)
    &\le \delta+23\psi\left(\frac{27}{4}\right)
        +\int_{24}^{\infty}\psi\left(\frac{135}{64}\log x\right)\,dx <0.481<\frac12,
\end{aligned}
\end{equation}
where we used $a_i=1$ for $2\le i\le24$ and integral comparison for the remaining decreasing terms.

Fix an ordinary component $S_j$, and order its pages locally as $q_1\le\cdots\le q_m$. If $\cert(S_j,3/4)>27\lambda/4$, then \eqref{eq:manpart-prefix-bound} gives $q_i/\lambda>(27/4)a_i$ for every $i\ge2$. Since the first page is also expensive, $q_1/\lambda>6$. The function $\psi$ is decreasing, so
\[
    \sum_{p\in S_j}u_p
    <\delta+\sum_{i=2}^{m}\psi\left(\frac{27}{4}a_i\right)
    <\frac12,
\]
contradicting \eqref{eq:manpart-component-coverage}.

For the residual component, the bound is immediate if $G=\emptyset$, since then $\cert_0=0$. Otherwise, order its normal pages as $q_1\le\cdots\le q_m$. The ghost comes first in the contracted ordering, so \eqref{eq:manpart-prefix-bound} gives $\cert_0\le q_i/a_{i+1}$ for every $i\ge1$. If $\cert_0>27\lambda/4$, then
\[
    \sum_{p\in S_0\setminus G}u_p
    <\sum_{i=1}^{m}\psi\left(\frac{27}{4}a_{i+1}\right)
    <\frac12-\delta,
\]
again contradicting \eqref{eq:manpart-component-coverage}. We have therefore shown that
\begin{equation}\label{eq:manpart-certificate-scale}
    \cert(S_j,3/4)\le\frac{27}{4}\lambda
    \quad (j\in[h]),
    \qquad
    \cert_0\le\frac{27}{4}\lambda.
\end{equation}

\paragraph{Charging costs to the global relaxation.}
We next show how the relaxation jointly pays for the ghost simulation and the remaining holes. More generally, suppose $\lambda$ is defined and a threshold $t\ge0$ separates $G$ from the other pages: $q_p\le t$ for $p\in G$ and $q_p\ge t$ for $p\in S\setminus G$. Using \eqref{eq:manpart-global-solution}, we obtain the identity
\begin{equation}\label{eq:manpart-budget}
\begin{aligned}
    \left(1+\frac{t}{\lambda}\right)\lopt(S,r)-Q_G-t(r-|G|)
    &=\sum_{p\in G}(t-q_p)(1-f_p^\star)
      +\sum_{p\in S\setminus G}(q_p-t)f_p^\star \ge0.
\end{aligned}
\end{equation}
Both sums are nonnegative because $t$ separates the two sets. In the present case, $G=C$, so taking $t=6\lambda$ gives
\[
    Q_G+6h\lambda\le7\lopt(S,r).
\]
Combining this budget with \eqref{eq:manpart-certificate-scale} and $h+1\le2h$, we conclude that
\[
\begin{aligned}
    \sum_{j=1}^{h}\cert(S_j,3/4)+Q_G+\cert_0
    &\le Q_G+\frac{27}{4}(h+1)\lambda\\
    &\le\frac94\left(Q_G+6h\lambda\right)\\
    &\le\frac{63}{4}\lopt(S,r)
    \le\frac{63}{4}\lman(S,r).
\end{aligned}
\]
as desired.

\subsection*{The case $h=0$.}
In that case, the algorithm returns at either Line~\hyperref[line:manpar-zero-returning]{\ref*{line:manpar-zero-returning}} or Line~\hyperref[line:manpar-cheap-return]{\ref*{line:manpar-cheap-return}}. In either case, $S_0=S$ and $|G|=r<|S|$, so feasibility is immediate. The simulation cost satisfies
\[
    Q_G\le7\lopt(S,r).
\]
Indeed, it is zero in the first return. In the second, $G$ consists of the $r$ cheapest pages, and $t:=\max_{p\in G}q_p\le6\lambda$ separates $G$ from the remaining pages. Applying \eqref{eq:manpart-budget} with $r-|G|=0$ gives the bound.

It remains to control $\cert_0$ using the ghost lower bound. Write $\ell:=\lghost(S,G)$, and let $(f_g,f_1,\ldots,f_m)$ be an optimal solution to its defining relaxation, with the normal pages ordered by $q_1\le\cdots\le q_m$. If $\cert_0=0$, there is nothing to prove. Otherwise, all normal pages have positive intrinsic cost, since a zero-cost normal page would give a zero-cost swap with the ghost. Thus, the ghost comes first in the contracted ordering, and the prefix and ghost-swap candidates give
\[
    q_i\ge \cert_0 a_{i+1},
    \qquad
    q_i\ge \cert_0\left(1+\frac{\mu_i}{\mu_g}\right)
    \qquad (i\in[m]).
\]
Multiplying these inequalities by $f_i$, summing, and using the balance constraint $\sum_{i=1}^m\mu_i f_i=\mu_gf_g$, we obtain
\begin{equation}\label{eq:manpart-ghost-fractions}
    \sum_{i=1}^{m}a_{i+1}f_i\le\frac{\ell}{\cert_0},
    \qquad
    f_g+\sum_{i=1}^{m}f_i\le\frac{\ell}{\cert_0}.
\end{equation}
The swap estimate controls the total miss fraction, including that of the free ghost. The prefix estimate additionally controls how cheaply these fractions can be spread across many normal pages. Together, they suffice to bound the coverage of the contracted instance.

Specifically, the elementary inequality $\Phi(f)\le xf+e^{-x}$ holds for every $f\in[0,1]$ and $x\ge0$. Applying it with $x=2$ to the ghost and $x=2+5a_{i+1}$ to normal page $i$, and then using \eqref{eq:manpart-ghost-fractions}, gives
\[
\begin{aligned}
    1
    &\le\Phi(f_g)+\sum_{i=1}^{m}\Phi(f_i)\\
    &\le2\left(f_g+\sum_{i=1}^{m}f_i\right)
       +5\sum_{i=1}^{m}a_{i+1}f_i
       +e^{-2}\left(1+\sum_{i=1}^{m}e^{-5a_{i+1}}\right)\\
    &\le\frac{7\ell}{\cert_0}
       +e^{-2}\left(1+\sum_{i=2}^{\infty}e^{-5a_i}\right).
\end{aligned}
\]
The remaining term is bounded by
\[
\begin{aligned}
    e^{-2}\left(1+\sum_{i=2}^{\infty}e^{-5a_i}\right)
    &\le e^{-2}\left(1+23e^{-5}+\int_{24}^{\infty}x^{-25/16}\,dx\right)\\
    &=e^{-2}\left(1+23e^{-5}+\frac{16}{9}24^{-9/16}\right)
    <0.197<\frac15.
\end{aligned}
\]
Consequently, $1<7\ell/\cert_0+1/5$, which yields $\cert_0\le(35/4)\lghost(S,G)$. Finally, since $|G|=r$,
\[
    Q_G+\cert_0
    \le7\lopt(S,r)+\frac{35}{4}\lghost(S,G)
    \le\frac{63}{4}\lman(S,r),
\]
again as desired.

\section{Supplemental Material for \Cref{sec:predictions} (Prediction Errors)}
In this chapter of the appendix, we present all omitted proofs and details from \Cref{sec:predictions}.

\subsection{A Simple Online Policy}\label{app:robust}

We give a deterministic, prediction-free policy with competitive ratio $O(\log(k+1))$. The policy groups the $r:=|S|-k$ pages of smallest intrinsic cost with one additional page, serves this group using a single cache slot, and pins all remaining pages. Formally:

\begin{algorithm}[h]
\caption{$\robustname(S,k)$}
\label{alg:robust}
\begin{algorithmic}[1]
\Statex \textbf{Input:} A non-empty set of pages $S$ and a cache size $k\in\{1,\ldots,|S|\}$.

\State Set $r:=|S|-k$.

\If{$r=0$}
    \State \textbf{Return} the policy that keeps every page in $S$ cached.
\EndIf

\State Let $G$ consist of the $r$ pages in $S$ with smallest intrinsic costs, and let $\mu_G:=\sum_{p\in G}\mu_p$.

\State Choose
\[
    b\in\arg\min_{p\in S\setminus G}
    \frac{\mu_Gq_p}{\mu_G+\mu_p}.
\]

\State Pin all pages in $S\setminus(G\cup\{b\})$ permanently in the cache.

\State Serve requests to $G\cup\{b\}$ using the remaining cache slot, always retaining the most recently requested page in this set.
\end{algorithmic}
\end{algorithm}

The policy uses $k-1$ slots for pinned pages and one slot for $G\cup\{b\}$, so it is feasible under mandatory admission. We next bound its cost rate.

\begin{lemma}
For every instance of $\pname$ with pages $S$ and cache size $k\ge1$,
\[
    \rate\bigl(\robustname(S,k)\bigr)
    \le4\log\bigl(2(k+1)\bigr)\cdot\rate(\opt).
\]
\end{lemma}

\begin{proof}
The case $k=|S|$ is immediate, so assume $r\ge1$. Only pages in $G\cup\{b\}$ can incur misses. For each $p\in G$, counting every request as a miss bounds its cost rate by $q_p$. After a miss on $b$, the policy retains $b$ until the next request to any page in $G$, and then leaves $b$ absent until its next request. These two waiting times have means $1/\mu_G$ and $1/\mu_b$, respectively. Hence, the cost rate due to $b$ is
\[
    \frac{w_b}{1/\mu_G+1/\mu_b}
    =\frac{\mu_Gq_b}{\mu_G+\mu_b}.
\]
Consequently,
\begin{equation}\label{eq:rob1}
    \rate\bigl(\robustname(S,k)\bigr)
    \le\sum_{p\in G}q_p
        +\frac{\mu_Gq_b}{\mu_G+\mu_b}.
\end{equation}
Write $\alpha:=\log(2(k+1))$. We bound the first term using $\lopt(S,r)$ and the second using the ghost lower bound $\lghost(S,G)$.

\paragraph{The Cost of the Cheap Pages.}
Let $\tau:=\max_{p\in G}q_p$. If $\tau=0$, then $\sum_{p\in G}q_p=0$. Otherwise, let $(f_p^\star)_{p\in S}$ be an optimal solution to $\lopt(S,r)$. Its feasibility gives $\sum_{p\in S}\Phi(f_p^\star)\ge r$, and its objective value is $\sum_{p\in S}q_pf_p^\star=\lopt(S,r)$. Applying the inequality $\Phi(f)\le sf+e^{-s}$ to each $f_p^\star$, with $s=\alpha q_p/\tau$, yields
\[
\begin{aligned}
    r
    &\le\sum_{p\in S}\Phi(f_p^\star)\le\sum_{p\in S}
        \left(\frac{\alpha q_p}{\tau}f_p^\star
        +e^{-\alpha q_p/\tau}\right) =\frac{\alpha}{\tau}\lopt(S,r)
        +\sum_{p\in S}e^{-\alpha q_p/\tau}.
\end{aligned}
\]
Rearranging, we obtain
\[
    \lopt(S,r)
    \ge\frac{\tau}{\alpha}
        \left(r-\sum_{p\in S}e^{-\alpha q_p/\tau}\right).
\]
Therefore, we have 
\begin{align*}
    \lopt(S,r)
    &\ge\frac{\tau}{\alpha}
        \left(r-\sum_{p\in S}e^{-\alpha q_p/\tau}\right)\\
    &=\frac{\tau}{\alpha}
        \left(
        \sum_{p\in G}(1-e^{-\alpha q_p/\tau})
        -\sum_{p\in S\setminus G}e^{-\alpha q_p/\tau}
        \right)\\
    &\ge\frac1\alpha
        \left(
        (1-e^{-\alpha})\sum_{p\in G}q_p
        -k\tau e^{-\alpha}
        \right)\\
    &\ge\frac{1-(k+1)e^{-\alpha}}{\alpha}
        \sum_{p\in G}q_p\\
    &=\frac1{2\alpha}\sum_{p\in G}q_p.
\end{align*}
The second line uses $|G|=r$. The third uses concavity of $1-e^{-\alpha x}$ on $[0,1]$ for pages in $G$, and $q_p\ge\tau$ for the $k$ pages outside $G$. The fourth uses $\tau\le\sum_{p\in G}q_p$, and the last uses $(k+1)e^{-\alpha}=1/2$. Thus, by \Cref{lemma:lb_man},
\begin{equation}\label{eq:rob2}
    \sum_{p\in G}q_p
    \le2\alpha\,\lopt(S,r)
    \le2\alpha\,\rate(\opt).
\end{equation}

\paragraph{The Cost of the Additional Page.}
Write
\[
    D:=\frac{\mu_Gq_b}{\mu_G+\mu_b}
      =\min_{p\in S\setminus G}
        \frac{\mu_Gq_p}{\mu_G+\mu_p}.
\]
Let $(f_g,(f_p)_{p\in S\setminus G})$ be an optimal solution to $\lghost(S,G)$, where $g$ is the zero-weight ghost with arrival rate $\mu_G$. By the definition of $D$, every normal page satisfies
\[
    q_p\ge D\left(1+\frac{\mu_p}{\mu_G}\right).
\]
Using the balance constraint $\mu_Gf_g=\sum_{p\in S\setminus G}\mu_pf_p$, we obtain
\begin{align*}
    \lghost(S,G)
    &=\sum_{p\in S\setminus G}q_pf_p\\
    &\ge D\sum_{p\in S\setminus G}
        \left(1+\frac{\mu_p}{\mu_G}\right)f_p\\
    &=D\left(f_g+\sum_{p\in S\setminus G}f_p\right).
\end{align*}
On the other hand, applying $\Phi(f)\le\alpha f+e^{-\alpha}$ to the covering constraint gives
\begin{align*}
    1
    &\le\Phi(f_g)+\sum_{p\in S\setminus G}\Phi(f_p)\\
    &\le\alpha\left(f_g+\sum_{p\in S\setminus G}f_p\right)
        +(k+1)e^{-\alpha}\\
    &=\alpha\left(f_g+\sum_{p\in S\setminus G}f_p\right)
        +\frac12.
\end{align*}
Therefore, $\lghost(S,G)\ge D/(2\alpha)$. Since $|G|=r$, \Cref{lemma:lb_man} implies
\begin{equation}\label{eq:rob3}
    \frac{\mu_Gq_b}{\mu_G+\mu_b}
    =D
    \le2\alpha\,\lghost(S,G)
    \le2\alpha\,\rate(\opt).
\end{equation}
Combining \eqref{eq:rob1}, \eqref{eq:rob2}, and \eqref{eq:rob3} yields
\[
    \rate\bigl(\robustname(S,k)\bigr)
    \le4\alpha\,\rate(\opt)
    =4\log\bigl(2(k+1)\bigr)\cdot\rate(\opt),
\]
as claimed.
\end{proof}

\subsection{Proof of \Cref{thm:cons}}\label{app:consistency}

We first propagate \Cref{lemma:guarded-eviction-errors} through the one-hole policy, and then apply the partition guarantee. Recall that $\cert(S,c)$ denotes the certificate returned by
$\ohp(S,c)$.

\begin{corollary}[One-Hole Policy with Prediction Errors]
\label{cor:one-hole-errors}
For any set $S$ of at least two pages, under mandatory admission and per-request predictions of total error $\error_S$,
\[
    \rate\bigl(\ohp(S,1/4)\bigr)
    \le (6+15\error_S)\cert(S,3/4).
\]
\end{corollary}

\begin{proof}
We first show that
\[
    \rate\bigl(\ohp(S,1/4)\bigr) \le (2+5\error_S)\cert(S,1/4).
\]
If the selected candidate is a swap, its rate equals its certificate independently of the predictions, so the bound is immediate. Otherwise, let $U$ be the selected prefix, $m:=|U|$, and $Q:=\max_{p\in U}q_p$. Its certificate satisfies
\[
    \cert(S,1/4)
    =\left(
        \frac{8}{\log m}
        +(1-m^{-1/4})^{m-2}
    \right)Q
    \ge \frac{8Q}{\log m}.
\]
Since the remaining pages are pinned, \Cref{lemma:guarded-eviction-errors} gives
\[
    \rate\bigl(\ohp(S,1/4)\bigr)
    \le (16+40\error_U)\frac{Q}{\log m}
    \le (2+5\error_S)\cert(S,1/4).
\]

Finally, changing $c$ from $3/4$ to $1/4$ triples the first term of each prefix certificate and decreases its second term. Swap certificates remain unchanged. Taking the minimum over the same collection of candidates therefore gives
\[
    \cert(S,1/4)\le 3\cert(S,3/4),
\]
which completes the proof.
\end{proof}

To apply \Cref{cor:one-hole-errors} to the contracted residual instance, we specify its ghost predictions and bound their error. At each ghost request at time $s$, first update the prediction of the requested original page. Writing $t'_p$ for the current next-arrival prediction of each $p\in G$, supply the ghost with the prediction
\[
    t'_g:=\max\left\{s,\min_{p\in G}t'_p\right\}.
\]
We claim that its error satisfies
\[
    \epsilon_g\le\sum_{p\in G}\epsilon_p.
\]

To see this, consider the next ghost request, at time $t$. If $t'_g>t$, the page requested at $t$ has a current prediction at least $t'_g$. Its overprediction error therefore covers $t'_g-t$, and each original arrival is charged at most once in this way. If $t'_g<t$, choose a page $p$ attaining the minimum predicted arrival time, and let $t_p\ge t$ be its actual next arrival. The ghost's error interval satisfies
\[
    [t'_g,t]\subseteq[t'_p,t_p].
\]
Moreover, clipping at $s$ ensures that this interval lies within the ghost interarrival interval $[s,t]$. Thus, repeated charges to the same original prediction occupy disjoint intervals, whose total length is at most its underprediction error $t_p-t'_p$.

Summing these charges proves the claimed rate bound. At a finite horizon, the unfinished prediction for each page contributes at most its elapsed interarrival time, which has bounded expectation. Consequently, the total prediction error in the contracted instance satisfies
\[
    \error_{\operatorname{Cont}(S_0,G)}
    =\epsilon_g+\sum_{p\in S_0\setminus G}\epsilon_p
    \le \error_{S_0}.
\]

\paragraph{Proof of \Cref{thm:cons}.}
The case $k\ge |S|$ is immediate by keeping every page cached. Otherwise, let $(\mathcal{T},S_0,G)$ be the output of $\manpar(S,|S|-k)$, with $\mathcal{T}=\{S_j\}_{j=1}^h$. Write $Q_G:=\sum_{p\in G}q_p$ and set
\[
    \cert_0:=
    \begin{cases}
        \cert(\operatorname{Cont}(S_0,G),3/4),&G\neq\emptyset,\\
        0,&G=\emptyset.
    \end{cases}
\]
Every ordinary component has error at most $\error$, and the same holds for the contracted residual instance by the preceding argument. Applying \Cref{cor:one-hole-errors} to these components and accounting for the additional simulation cost $Q_G$ gives
\[
\begin{aligned}
    \rate\bigl(\manfinal(S,|S|-k,1/4)\bigr)
    &\le
    (6+15\error)
    \left(\sum_{j=1}^h \cert(S_j,3/4)+\cert_0\right)+Q_G\\
    &\le
    (6+15\error)
    \left(\sum_{j=1}^h \cert(S_j,3/4)+Q_G+\cert_0\right)\\
    &\le
    15.75(6+15\error)\lman(S,|S|-k)\\
    &\le
    15.75(6+15\error)\rate(\opt).
\end{aligned}
\]
Here the third inequality uses \Cref{lem:mandatory-partition-feasibility}, and the last uses \Cref{lemma:lb_man}. This establishes the claimed $O(1+\error)$ competitive ratio.

\newpage

\end{appendices}



\newpage
\bibliographystyle{plainnat}
\bibliography{references}
\end{document}